\documentclass[11pt]{article}
\usepackage{amsmath}
\usepackage{amsthm}
\usepackage{amssymb}
\usepackage{algorithm}
\usepackage{subfig}
\usepackage{color}
\usepackage[english]{babel}
\usepackage{graphicx}
\usepackage{wrapfig,epsfig}
\usepackage{epstopdf}
\usepackage{url}
\usepackage{graphicx}
\usepackage{color}
\usepackage{epstopdf}
\usepackage{algpseudocode}
\usepackage{scrextend}
\usepackage[T1]{fontenc}
\usepackage{bbm}
\usepackage{comment}

\usepackage{tablefootnote}
\usepackage[flushleft]{threeparttable}

\usepackage{tikz}
\usepackage{hyperref}  
\hypersetup{colorlinks=true,citecolor=blue,linkcolor=blue} 
\usetikzlibrary{arrows}
\usepackage[margin=1in]{geometry}
\graphicspath{{./figs/}}

\author{
  Yin Tat Lee\thanks{
  \texttt{yintat@uw.edu}
  University of Washington \& Microsoft Research}
  \and
  Zhao Song\thanks{
  \texttt{zhaos@utexas.edu}
  UT-Austin \& University of Washington}
  \and
  Santosh S. Vempala \thanks{
  \texttt{vempala@gatech.edu}
  Georgia Tech}
}

\date{}
\title{Algorithmic Theory of ODEs and Sampling from Well-conditioned Logconcave Densities}
\newtheorem{theorem}{Theorem}[section]
\newtheorem{lemma}[theorem]{Lemma}
\newtheorem{definition}[theorem]{Definition}

\newtheorem{corollary}[theorem]{Corollary}

\newtheorem{remark}[theorem]{Remark}
\newtheorem{claim}[theorem]{Claim}

\newcommand{\wt}{\widetilde}
\newcommand{\ov}{\overline}
\newcommand{\eps}{\epsilon}
\newcommand{\defeq}{\overset{\mathrm{def}}{=}}

\newcommand{\R}{\mathbb{R}}

\renewcommand{\d}{\mathrm{d}}

\renewcommand{\varepsilon}{\epsilon}
\renewcommand{\tilde}{\wt}

\renewcommand{\i}{\mathbf{i}}

\DeclareMathOperator*{\E}{{\bf {E}}}

\DeclareMathOperator{\T}{\mathcal{T}}

\definecolor{darkgreen}{RGB}{0,100,0}
\definecolor{mygreen}{RGB}{80,180,0}
\definecolor{b2}{RGB}{51,153,255}
\definecolor{mycy2}{RGB}{255,51,255}

\makeatletter
\newcommand*{\RN}[1]{\expandafter\@slowromancap\romannumeral #1@}
\makeatother

\usepackage{lineno}

\usepackage{etoolbox}

\makeatletter

\newcommand{\define}[4][ignore]{%
  \ifstrequal{#1}{ignore}{}{
  \@namedef{thmtitle@#2}{#1}}%
  \@namedef{thm@#2}{#4}%
  \@namedef{thmtypen@#2}{lemma}%
  \newtheorem{thmtype@#2}[theorem]{#3}%
  \newtheorem*{thmtypealt@#2}{#3~\ref{#2}}%
}

\newcommand{\state}[1]{%
  \@namedef{curthm}{#1}
  \@ifundefined{thmtitle@#1}{
  \begin{thmtype@#1}
    }{
  \begin{thmtype@#1}[\@nameuse{thmtitle@#1}]
  }
    \label{#1}
    \@nameuse{thm@#1}
  \end{thmtype@#1}
  \@ifundefined{thmdone@#1}{
  \@namedef{thmdone@#1}{stated}%
  }{}
}

\newcommand{\restate}[1]{%
  \@namedef{curthm}{#1}
  \@ifundefined{thmtitle@#1}{
    \begin{thmtypealt@#1}
    }{
  \begin{thmtypealt@#1}[\@nameuse{thmtitle@#1}]
  }
    \@nameuse{thm@#1}
  \end{thmtypealt@#1}
  \@ifundefined{thmdone@#1}{
  \@namedef{thmdone@#1}{stated}%
  }{}
}

\newcommand{\thmlabel}[1]{
  \@ifundefined{thmdone@\@nameuse{curthm}}{\label{#1}
    }{\tag*{\eqref{#1}}}
}
\makeatother

\begin{document}

\begin{titlepage}
  \maketitle
\begin{abstract}
Sampling logconcave functions arising in statistics and machine learning has been a subject of intensive study. Recent developments include analyses for Langevin dynamics and Hamiltonian Monte Carlo (HMC). While both approaches have dimension-independent bounds for the underlying {\em continuous} processes under sufficiently strong smoothness conditions, the resulting discrete algorithms have complexity and number of function evaluations growing with the dimension. Motivated by this problem, in this paper, we give a general algorithm for solving multivariate ordinary differential equations whose solution is close to the span of a known basis of functions (e.g., polynomials or piecewise polynomials). The resulting algorithm has polylogarithmic depth and essentially tight runtime --- it is nearly linear in the size of the representation of the solution.

We apply this to the sampling problem to obtain  
a nearly linear implementation of HMC for a broad class of smooth, strongly logconcave densities, with the number of iterations (parallel depth) and gradient evaluations being {\em polylogarithmic} in the dimension (rather than polynomial as in previous work). This class includes the widely-used loss function for logistic regression with incoherent weight matrices and has been subject of much study recently. We also give a faster algorithm with {\em polylogarithmic depth} for the more general and standard class of strongly convex functions with Lipschitz gradient. 
These results are based on (1) an improved contraction bound for the exact HMC process and (2) logarithmic bounds on the degree of polynomials that approximate solutions of the differential equations arising in implementing HMC.


  \end{abstract}
  \thispagestyle{empty}
\end{titlepage}

{\hypersetup{linkcolor=black}
\tableofcontents
}
\newpage


\section{Introduction}

The complexity of sampling
a high-dimensional density of the form $e^{-f(x)}$ where $f$ is a convex function is a fundamental problem with many applications \cite{ls90,ls92,ls93,lv06a,lv06b,d17,dk17,dr18,dcwy18}. The focus of this paper is to give very fast, i.e., nearly linear time algorithms, for a large subclass of such densities. A motivating and important case is the loss function for logistic regression, widely used in machine learning applications \cite{b44,p00,nj02,hls13,b14}:
\[
\sum_{i=1}^n\phi_i(a_i^\top x)
\]
where $\phi_i$ are convex functions; a popular choice is $\phi(t)=\log(1+e^{-t})$. Sampling according to $e^{-f}$ for this choice of $f$ corresponds to sampling models according to their KL-divergence, a natural and effective choice for classification problems \cite{hls13}. 

A general approach to sampling is by an ergodic Markov chain whose stationary distribution is designed to have the desired density. Traditionally, this is done via a Metropolis filter, which accepts a proposed (random) next step $y$ from the current point $x$ with probability $\min \{1,\frac{f(y)}{f(x)}\}$. While very general, one downside of this approach is the possibility of high rejection probabilities, which typically force local steps to be very small. Nevertheless, for arbitrary logconcave functions (including nonsmooth ones), this approach has the current best guarantees \cite{lv06b}. 

Another family of algorithms is derived from an underlying {\em continuous} stochastic process with the desired stationary density. A classic example of such a continuous process is Brownian motion. To sample a convex body for example, one could use Brownian motion with a boundary reflection condition. This is written as the stochastic equation: 
\[
\d X_t = \d W_t 
\]
with reflection at the boundary of the domain, and $dW_t$ being infinitesimal Brownian motion. To sample from the density proportional to $e^{-f(x)}$, one can use the stochastic differential equation,
\[
\d X_t = - \nabla f(X_t) \d t + \sqrt{2} \d W_t.
\]
By the classical Fokker-Planck equation, under mild assumptions on $f$, the stationary density of this process is proportional to $e^{-f(x)}$.

How can we turn these continuous processes into algorithms? One approach is to take small rather than infinitesimal steps, and this leads to the Langevin dynamics, of which there are multiple flavors \cite{d17,dk17,zlc17,rrt17,dr18,ccbj18,ccabj18,cfmbj18}. Starting with Dalalyan \cite{d17}, it has been established that these dynamics converge in polynomial (in dimension) time for strongly logconcave functions, with the underdamped version converging in $O(\sqrt{d})$ iterations (and polynomial dependences on appropriate condition numbers) \cite{ccbj18}. The dependence on dimension seems unavoidable in the discretized algorithm, even though the continuous process has no such dependence.

\paragraph{Hamiltonian Monte Carlo.} HMC is a random process that maintains a position $x$ and velocity pair $v$. To sample according to $e^{-f}$, we define a Hamiltonian $H(x,v) = f(x)+\frac{1}{2} \|v\|^2$. At each step $v$ is chosen randomly from $N(0,I)$ and $x$ is updated using the following Ordinary Differential Equation (ODE) for a some fixed time interval.
\[
\frac{\d x(t)}{\d t} = v(t), \quad \frac{\d v(t)}{\d t} = -\nabla f( x(t) ).
\]
This process has the particularly nice property that it conserves the value of $H$, and as a result there is no need to apply a Metropolis filter. 
HMC has been studied in many works \cite{ms17,mv18,lv18}.
Mangoubi and Smith \cite{ms17} gave the following guarantee for strongly logconcave densities.

\begin{theorem}[\cite{ms17}]
Let $f$ be a smooth, strongly convex function s.t. for all $y$
$$
m_2 \cdot I \preceq \nabla^2 f(y) \preceq M_2 \cdot  I.
$$
Then, 
HMC converges to the density proportional to $e^{-f}$ in $\wt{O}((M_2/m_2)^2)$ iterations with each iteration being the exact solution of an ODE.
\end{theorem}
For the resulting algorithm presented in \cite{ms17}, which needs to approximate the solution of the ODE, the number of function evaluations and overall time grow as square-root of the dimension (and a higher polynomial of the condition number). Table \ref{table:relwork} summarizes related work on sampling logconcave functions with various structural assumptions. In all these cases, even with higher-order smoothness and incoherence assumptions, the number of gradient/function evaluations grows as a polynomial in $d$. A special case of much interest is Bayesian logistic regression. To address this \cite{mv18} define an incoherence parameter and achieve the previously best dependence on the dimension of $d^{1/4}$ for functions with bounded incoherence. They note that this is nearly optimal for the leapfrog implementation of HMC they use. Improving the complexity further, and in particular the dependence on the dimension $d$ is an important open problem. This brings us to our main motivating question: 

{\em For what class of functions can we avoid polynomial dependence on dimension (in an algorithm)? Can we do this for the logistic loss function?}

\newcommand{\cold}[1]{{\color{darkgreen} #1}}
\begin{table}[t]

\begin{center}
\begin{threeparttable}
    \begin{tabular}{ | l | l | l | l | l |}
    \hline
    {\bf method} & {\bf $\#$ iterations/} & {\bf $\#$ gradients} & {\bf total time} & {\bf reference}\\ 
    & {\bf parallel depth} & {\bf per iteration} & &\\ \hline
    Ball Walk/Hit-and-run\ifdefined\ourfs \tnote{$*$} \else \tablefootnote{\label{foo:coldwarm} have different bounds for warm start and general (cold) start. \cold{We stated the runtime for cold start in green color.}} \fi & $d^3$\cold{, $d^4$} & $1$ & $d^5$\cold{, $d^6$} &\cite{lv06b} \\  \hline 
    LMC\ifdefined\ourfs \tnote{$*$} \else \footref{foo:coldwarm} \fi & $\kappa^2d$\cold{, $\kappa^3 d^3$} & $1$ & $\kappa^2 d^2$\cold{, $\kappa^3 d^4$}  & \cite{d17} \\ \hline 
    LMCO\ifdefined\ourfs \tnote{$*$} \else \footref{foo:coldwarm} \fi & $\kappa^2 d$\cold{, $\kappa^2 d^{2.5}$} & $1$ & $\kappa^2 d^4$\cold{, $\kappa^2 d^{5.5}$} & \cite{d17}  \\ \hline  
    Damped Langevin & $ \kappa^2 d^{0.5}$ & $1$ & $\kappa^2 d^{1.5}$ & \cite{ccbj18} \\ \hline 
    MALA\ifdefined\ourfs \tnote{$*$} \else \footref{foo:coldwarm} \fi &  $\kappa d$\cold{, $\kappa d^2$} & $1$ & $\kappa d^{2}$\cold{, $\kappa d^3$}  & \cite{dcwy18} \\ \hline 
    HMC & $\kappa^{6.5} d^{0.5}$ & $1$ & $\kappa^{6.5} d^{1.5}$  & \cite{ms17}\\ \hline 
    HMC\ifdefined\ourfs \tnote{$*$,$\dagger$} \else \footref{foo:coldwarm}$^,$\tablefootnote{\label{foo:smooth} make smoothness and incoherence assumptions motivated by and applicable to Bayesian logistic regression.} \fi & $\kappa^{2.75} d^{0.25}$\cold{, $\kappa^{3.5} d^{0.25}$} & $1$ & $\kappa^{2.75} d^{1.25}$\cold{, $\kappa^{3.5} d^{1.25}$}  & \cite{mv18}\\ \hline 
    HMC\ifdefined\ourfs\tnote{$\dagger$} \else\footref{foo:smooth}\fi & $\kappa^{1.5}$ & $1$ & $\kappa^{1.5} d$ & This paper\\ 
     HMC & $\kappa^{1.5}$ & $\kappa^{0.25} d^{0.5}$ & $\kappa^{1.75} d^{1.5}$  & \\ \hline
    \end{tabular}
    \ifdefined\ourfs 
    \begin{tablenotes}
    \item[$*$] have different bounds for warm start and general (cold) start. \cold{We stated the runtime for cold start in green color.}
    \item[$\dagger$] make smoothness and incoherence assumptions motivated by and applicable to Bayesian logistic regression.
  \end{tablenotes}
  \else \fi
\end{threeparttable}
\end{center}
\caption{Summary of results, $d$ is the dimension, $\kappa$ is the condition number of $\nabla^2 f$. We use the parallel depth of the algorithm as the number of iterations. We suppress polylogarithmic terms and dependence on the error parameter. Ball walk/hit-and-run apply to general logconcave distributions, the rest assume strongly logconcave with Lipschitz gradient and possibly more.  
In all previous work, for simplicity, we report the most favorable bounds by making various assumptions such as $\kappa \ll d$.}
\label{table:relwork}
\end{table}

\subsection{Results}

We begin with an informal statement of our result for sampling from a class that includes the logistic loss function. 


\begin{theorem}[Informal version of Theorem~\ref{thm:sampling_formal}]\label{thm:sampling_informal}
Let $A=[a_{1};a_{2};\cdots;a_{n}]\in\R^{n\times d}$, $\phi_{i}:\R\rightarrow\R$ with its $k$-th derivatives bounded by $O(1)^k$
and 
\[
f(x)=\sum_{i=1}^{n}\phi_{i}(a_{i}^{\top}x)+\frac{m_{2}}{2}\|x\|^{2}.
\]
Suppose that $\nabla^2 f$ has condition number $\kappa$ and $\tau=\|AA^{\top}\|_{\infty\rightarrow\infty}$.
Then we can find a random
point $X$ whose Wasserstein distance to $Y$ drawn from the density proportional to $e^{-f}$ satisfies
\[
W_2(X,Y) \leq \frac{\epsilon}{\sqrt{m_{2}}}
\]
using $\wt{O}( \kappa^{1.5} + \frac{\tau}{m_2})$ iterations, where 
each iteration takes
$\wt{O}(d)$ time and $\wt{O}(1)$ 
evaluations of $\nabla f$. 
\end{theorem}

\begin{remark}
The $\frac{1}{\sqrt{m_2}}$ term in the error is needed to make the statement invariant under scaling of $f$.
\end{remark}

For the logistic loss\footnote{The logistic function is $g(t) = \frac{1}{1+e^{-t}}$ and the logistic loss is $-\log ( g(t) ) = \log ( 1 + e^{-t} )$.} $\phi(t) = \log (1 + e^{-t})$, we have $\phi'(t) = - \frac{1}{1+e^t}$, and it has Cauchy estimate $M=1$ with radius $r=1$ (See Lemma~\ref{lem:property_of_logisitc_function}). 
The above result has the following application,
\begin{corollary}[Logistic loss sampling]
Let $f(x)=\sum_{i=1}^{n}\phi_{i}(a_{i}^{\top}x)+\frac{m_{2}}{2}\|x\|^{2}$
with $\phi(t)=\log(1+e^{-t})$. Let $\tau=\|AA^{\top}\|_{\infty\rightarrow\infty}$
and suppose that $\nabla^{2}f(x)\preceq M_{2}\cdot I$ for all $x$.
Starting at the minimum $x^{(0)}$ of $f$, we can find a random
point $X$ whose Wasserstein distance to $Y$ drawn from the density proportional to $e^{-f}$ satisfies
\[
W_2(X,Y) \leq \frac{\epsilon}{\sqrt{m_{2}}}
\]
using $\widetilde{O}(\frac{\tau}{m_{2}}+\kappa^{1.5})$ iterations
with $\kappa=\frac{M_{2}}{m_{2}}$. Each iteration takes $\widetilde{O}(d)$
time and $\widetilde{O}(1)$ matrix-vector multiplications for $A$ and $A^\top$.
\end{corollary}

\begin{remark}
Lemma \ref{lem:tau} shows that $\|A A^\top \|_{\infty\rightarrow\infty} = \Theta( \lambda_{\max}(AA^\top))$ for sparse enough matrix $A A^\top$. Since $\lambda_{\max}(AA^\top)$ usually has the same order as $M_2$, the number of iterations is dominated by the $\kappa^{1.5}$ term.
\end{remark}

The above results extend and improve previous work substantially.
First, in all previous algorithms, the number of functions calls was polynomial in the dimension $d$, while the dependence here is polylogarithmic. Second, our incoherence assumption for logistic regression is simpler and milder. Third, due to the nature of how we implement each step, the parallel depth of the algorithm is just the number of iterations, i.e., polylogarithmic in the dimension and $\tilde{O}(1)$ when the condition numbers are bounded. Fourth, the runtime and depth of our algorithm depends polynomially in $\log(1/\epsilon)$ while all previous (nearly) linear time algorithms  depends polynomially in $1/\epsilon$

We also give an improved bound on the complexity of sampling from $e^{-f}$ when $f$ is strongly convex and has a Lipschitz gradient (no further smoothness assumptions). 
\define{thm:strongly_convex}{Theorem}{\rm{(Strongly Convex){\bf.}}
Given a function $f$ such that $0\prec m_{2}\cdot I\preceq\nabla^{2}f(x)\preceq M_{2}\cdot I$
for all $x\in\mathbb{R}^{d}$ and $0<\epsilon<\sqrt{d}$. Starting $x^{(0)}$
at the minimum of $f$, we can find a random
point $X$ whose Wasserstein distance to $Y$ drawn from the density proportional to $e^{-f}$ satisfies
\[
W_2(X,Y)\leq\frac{\epsilon}{\sqrt{m_{2}}}
\]
using $O(\kappa^{1.5}\log(\frac{d}{\epsilon}))$ iterations where $\kappa=\frac{M_{2}}{m_{2}}$.
Each iteration takes $O\left(\frac{\kappa^{\frac{1}{4}}d^{\frac{3}{2}}}{\epsilon}\log\left(\frac{\kappa d}{\epsilon}\right)\right)$
time and $O\left(\frac{\kappa^{\frac{1}{4}}d^{\frac{1}{2}}}{\epsilon}\log\left(\frac{\kappa d}{\epsilon}\right)\right)$
evaluations of $\nabla f$, amortized over all iterations. 

}
\state{thm:strongly_convex}
The previous best bound was $\kappa^2 \sqrt{d}$ iterations \cite{ccbj18}. This result is one of the key surprises of this paper. Although this problem has been studied extensively with specifically-designed algorithms and analysis, we show how to get a better result by a general ODE algorithm and a general analysis which works for any ODE. Furthermore, {\em our algorithm is the first to achieve polylogarithmic depth dependence on the dimension, which seemed impossible in prior work.} 

The above results are based on three ingredients: (1) a new contraction rate for HMC of $\kappa^{1.5}$, improving on the previous best bound of $\kappa^2$ (2) a proof that a solution to ODE's arising from HMC applied to the above problem are approximated by (piecewise) low-degree polynomials and (3) a fast (nearly linear time and polylog parallel depth) algorithm for solving multivariate second-order ODEs.

We next present the multivariate high-order ODE guarantee. This generalizes and improves on the guarantee from \cite{lv17a}. While we state it below for the case of the piecewise polynomial basis of functions, it applies to any basis of functions. This is a general result about solving ODE efficiently, independent of the application to sampling. The only assumptions needed are that the ODE function is Lipschitz and that the solution is close to the span of small number of basis of functions. These natural assumptions suffice to get around the worst-case complexity lower bounds for solving such general ODEs \cite{kf82,ko83,ko10,ka10,kc12}.

\begin{theorem}[Informal version of Theorem~\ref{thm:kth_order_ode_piecewise} for 1st order ODE]
Let $x^*(t) \in \R^d$ be the solution of the ODE 
\begin{align*}
    \frac{\d }{\d t} x(t) = F(x(t), t), x(0) = v
\end{align*}
where $F : \R^{d+1} \rightarrow \R^{d}$, $x(t) \in \R^d$ and $v\in \R^d$. Given some $L$ and $\epsilon >0$ such that 

1. There exists a piece-wise polynomial $q(t)$ such that $q(t)$ on $[T_{j-1},T_j]$ is a degree $D_j$ polynomial with $$0 = T_0 < T_1 < \cdots < T_n = T$$ and that
\begin{align*}
\left\| q(t) - \frac{ \d }{ \d t } x^*(t) \right\| \leq \frac{ \epsilon }{ T }, \forall t \in [0,T]
\end{align*}

2. The algorithm knows about the intervals $[T_{j-1},T_j]$ and the degree $D_j$ for all $j \in [n]$.

3. For any $y , z \in \R^{d}$,
\begin{align*}
\| F( y, t ) - F( z , t) \| \leq  L \| y - z \|, \forall t \in [0,T].
\end{align*}
Assume $L T \leq 1/16000$. Then, we can find a piece-wise polynomial $x(t)$ such that
\begin{align*}
\max_{t \in [0,T] } \left\|    x(t) -  x^*(t) \right\| \lesssim \epsilon.
\end{align*}
using $\tilde{O}\left( \sum_{ i = 1 }^n (1 + D_i) \right)$ evaluations of $F$ and
$\tilde{O}\left(d \sum_{ i = 1 }^n (1+D_i)\right)$ time.
\end{theorem}


We suspect these methods will be useful in many other settings beyond the focus application of this paper. 
Moreover, the result is nearly optimal. Roughly speaking, it says that the complexity of solving the ODE is nearly the same as the complexity of representing the solution. The assumption that $LT < 1$ is essential, as otherwise after longer time, the solution can blow up exponentially. Also, the assumption on the piece-wise polynomial approximation has a certain universality since this is how one implicitly represents a function using any iterative method. Finally, each iteration of the ODE algorithm, the collocation method, can be fully parallelized; as a result the parallel time complexity of the sampling algorithms in this paper are polylogarithmic in the dimension.  


\subsection{HMC and improved contraction rate}
We give an improved contraction rate for HMC, stated explicitly as Algorithm \ref{alg:hmc}.
\begin{algorithm}\caption{Hamiltonian Monte Carlo Algorithm}\label{alg:hmc}
\begin{algorithmic}[1]
\Procedure{HMC}{$x^{(0)},f,\epsilon, h$} \Comment{Theorem~\ref{thm:contraction_HMC}} 
\State Suppose that $f$ is $m_2$ strongly convex with $M_2$ Lipschitz gradient
on $\mathbb{R}^{d}$.
\State Assume that the step size $h \leq \frac{ m_2^{1/4}}{2 M_2^{3/4}}$.
\State Let the number of iterations $N = \frac{1}{\theta} \cdot\log\left(\frac{4}{\epsilon^2} \left(\frac{\|\nabla f(x^{(0)})\|_{2}^{2}}{ m_2}+d \right)\right)$
with $\theta = \frac{ m_2 h^2 }{ 8 }$.
\For { $k=1,2,\cdots,N$}
	\State Generate a Gaussian random direction $v\sim {\cal N}(0,I_{d})$.
	\State Let $x(t)$ be the HMC defined by
	\[
\frac{\d^{2}x}{\d t^{2}}=-\nabla f(x),\frac{\d x}{\d t}(0)=v,x(0)=x^{(k-1)}.
\]
	\State Find a point $x^{(k)}$ such that \Comment{Theorem~\ref{thm:kth_order_ode_piecewise}}
	\[ \|x^{(k)}-x(h)\|_{2}\leq \overline{\epsilon}:=\frac{\theta \cdot \epsilon}{2 \sqrt{m_2}}.\]
\EndFor
\State \Return $x^{(N)}$.
\EndProcedure
\end{algorithmic}
\end{algorithm}
We give two contraction
bounds for the ideal HMC. The first bound is $(\frac{m_2}{M_2})^{1.5}$
using $T\sim\frac{m_2^{1/4}}{M_2^{3/4}}$. The second bound shows that
there is a $T$ that gives the optimal contraction bound $\frac{m_2}{M_2}$. However, as we will see
this part cannot be used to bound the overall mixing time, because the time $T$ depends on the point we use
for coupling, which is unknown to the algorithm. We keep this as evidence for a possible $\frac{m_2}{M_2}$ bound. 
The improvement is from $\kappa^2$ in previous work to $\kappa^{1.5}$.

\begin{lemma}[Contraction bound for HMC]\label{lem:HMC_contraction}
Let $x(t)$ and $y(t)$ be the solution
of HMC dynamics on $e^{-f}$ starts at $x(0)$ and $y(0)$ with initial
direction $x'(0)=y'(0)=v$ for some vector $v$. Suppose that $f$
is $m_2$ strongly convex with $M_2$ Lipschitz gradient., i.e., $m_2 \cdot I\preceq\nabla^{2}f(x)\preceq M_2 \cdot I$
for all $x$. Then, for $0\leq t\leq\frac{m_2^{1/4}}{2M_2^{3/4}}$, we
have that
\[
\|x(t)-y(t)\|_{2}^{2}\leq\left(1-\frac{ m_2 }{4}t^{2}\right)\|x(0)-y(0)\|_{2}^{2}.
\]
Furthermore, there is $t\geq0$ depending on $f, x(0), y(0)$, and $v$
such that 
\[
\|x(t)-y(t)\|_{2}^{2}\leq\left(1-\frac{1}{16}\frac{ m_2 }{ M_2 }\right)\|x(0)-y(0)\|_{2}^{2}.
\]
\end{lemma}

\subsection{Techniques}

In this paper we give bounds on the Collocation Method for solving ODEs. To ensure the algorithm is applicable to the ODE's that arise in the sampling application, we need to show that the solution of the ODE is close to a low-rank basis. Given only bounds on the Hessian of a function, we do this by approximating the solution of the ODE with a piecewise degree two polynomial. For smooth functions, we can use low-degree polynomials.  

The proofs of the degree bounds go via the Cauchy-Kowalevsky method, by showing bounds on all derivatives at the initial point. To do this for multivariate ODE's, we use the method of majorants, and reduce the problem to bounding the radius of convergence of one-variable ODEs. 

The improved convergence guarantees for exact HMC are also based on better analysis of the underlying ODE, showing that a larger step size than previously known is possible.   

For many optimization and sampling methods, there are corresponding customized versions that deal with decomposable functions by sampling terms of the functions. These algorithms usually take nearly linear time with the number of iterations being polynomial in the dimension. Often, an improvement in the general case would lead to an improvement in the decomposable case. To limit the length of this paper, we focus only on results with polylogarithmic depth. Therefore, in Table \ref{table:relwork}, we do not list algorithms that involve sampling each term in decomposable functions \cite{bfr16,dsmbr16,drwpsx16,bffn17,dk17,cwzsc17,ndhvsz17,cfmbj18}. We expect our techniques can be further improved for decomposable functions by sampling individual terms. 

\paragraph{Outline of paper.}
 Then we give the main ODE algorithm and guarantees in Section \ref{sec:ode}. We give the proof of the improved convergence bound for HMC in Section \ref{sec:contraction}. We use both parts to obtain improved guarantees for sampling strongly logconcave functions with Lipschitz gradients in Section \ref{sec:nonsmooth} and smooth functions, including logistic loss in Section \ref{sec:logistic}.

Some preliminaries including standard definitions and well-known theorems about ODEs are in an appendix.
Remaining proofs about ODEs are in Appendix~\ref{sec:app_ode}. In Appendix~\ref{sec:app_logistic}, we present some useful tools for Cauchy estimates. Appendix~\ref{sec:app_cauchy} shows how to calculate the Cauchy estimates of some function which are extensively used in practice.

\section{ODE Solver for any basis}\label{sec:ode}
In this section, we analyze the collocation method for solving ODEs \cite{i09}. This method is classical in numerical analysis. The goal of this section is provide an introduction of this method and provide a non-asymptotic bounds for this method. In \cite{lv17a, lv18}, we applied this method to obtain faster algorithms for sampling on polytopes. Unfortunately, the particular version of collocation method we used assume the solution can be approximated by a low-degree polynomial, which heavily restrict the set of functions we can sample.

To give an intuition for the collocation method, we first consider the following first-order ODE
\begin{align*}
\frac{ \d }{ \d t } x(t) = & ~ F( x(t), t ), \quad \forall 0 \leq t \leq T ,  \\
x(0) = & ~ v.
\end{align*}
where $ F : \R^{d+1} \rightarrow \R$. 
The collocation method is partly inspired by the Picard-Lindel\"{o}f theorem, a constructive existence proof for a large class of ODE. Therefore, we will first revisit the proof of Picard-Lindel\"{o}f theorem for first-order ODE.

\subsection{Picard-Lindel\"{o}f theorem}\label{sec:picard_lindelof}
In general, we can rewritten the first-order ODE as an integral equation
$$x(t) = v + \int^t_0 F(x(s),s) \d s\quad \text{for all } 0 \leq t \leq T.$$

To simplify the notation, we use $\mathcal{C}([0,T],\R^d)$ to denote $\R^d$-valued functions on $[0,T]$. We define the operator $\T$ from $\mathcal{C}([0,T],\R^d)$ to $\mathcal{C}([0,T],\R^d)$ by
\begin{equation}
\T(x) (t) = v + \int_0^t F( x(s), s ) \d s \quad \text{for all } 0\leq t \leq T . \label{eq:def_T_operator}
\end{equation}
Therefore, the integral equation is simply $x = \T(x)$.

Banach fixed point theorem shows that the integral equation $x = \T(x)$ has a unique solution if there is a norm, and $j \in\mathbb{N}$ such that the map $\T^{\circ j}$ has Lipschitz constant less than 1. Recall that $\T^{\circ j}$ is the composition of $j$ many $\T$, i.e., 
\begin{align*}
\T^{\circ j}(x) = \underbrace{ \T ( \T ( \cdots \T }_{j~\text{many}~\T} ( x) \cdots ) ) .
\end{align*}

Picard-Lindel\"{o}f theorem shows that if $F$ is Lipschitz in $x$, then the map $\T^{\circ j}$ has Lipschitz constant less than 1 for some positive integer $j$.
\begin{lemma}\label{lem:Lipschitz_T}
Given any norm $\| \cdot \|$ on $\R^d$. Let $L$ be the Lipschitz constant of $F$ in $x$, i.e. 
$$\|F(x,s)-F(y,s)\| \leq L \|x-y\| \quad \text{for all } x, y\in \R^d, s\in[0,T].$$

For any $x\in\mathcal{C}([0,T],\R^d)$, we define the corresponding norm
$$\| x\| \defeq \max_{0 \leq t \leq T} \|x(t)\|.$$
Then, the Lipschitz constant of $\T^{\circ j}$ in this norm is upper bounded by $(LT)^j/j!$.
\end{lemma}
\begin{proof}
We prove this lemma with a stronger induction statement 
$$\| (\T^{\circ j} x)(h) -  (\T^{\circ j} y)(h) \|\leq \frac{ L^{j} h^{j} }{ j ! } \| x - y \| \quad \text{for all } 0\leq h \leq T.$$

The base case $j=0$ is trivial. For the induction case $j$, we can upper bound the term as follows
\begin{align*}
& ~ \| \T^{\circ j} x (h) - \T^{\circ j} y (h) \| \\
= & ~ \| \T \T^{\circ (j-1)} x (h)- \T \T^{\circ (j-1)} y (h) \| \\
\leq & ~ \int_0^h \| F( \T^{\circ (j-1)} x(s), s ) - F( \T^{\circ (j-1)} y(s), s ) \| \d s \\
\leq & ~ L  \int_0^h \| \T^{\circ (j-1)} x(s) - \T^{\circ (j-1)} y(s) \| \d s & \text{~by~}f\text{~is~$L$~Lipschitz}\\
\leq & ~ L \int_0^h \frac{ L^{j-1} s^{j-1} }{ (j-1) ! } \| x - y \| \d s  & \text{ by the induction statement}\\
= & ~ \frac{L^j}{j!} h^j \| x - y \|.
\end{align*}
This completes the induction.
\end{proof}

\subsection{Intuition of collocation method}
To make the Picard-Lindel\"{o}f theorem algorithmic,
we need to discuss how to represent a function in $\mathcal{C}([0,T],\R^{d})$.
One standard way is to use a polynomial $p_{i}(t)$ in $t$ for each
coordinate $i\in[d]$.  In this section, we assume that there is
a basis $\{\varphi_{j}\}_{j=1}^{D}\subset\mathcal{C}([0,T],\R)$ such
that for all $i\in[d]$, $\frac{\d x_{i}}{\d t}$ is approximated by some
linear combination of $\varphi_{j}(t)$. 

For example, if $\varphi_{j}(t)=t^{j-1}$ for $j\in [d]$, then our assumption is
simply saying $\frac{\d x_{i}}{\d t}$ is approximated by a degree $D-1$
polynomial. Other possible basis are piecewise-polynomial and Fourier
series. By Gram-Schmidt orthogonalization, we can always pick nodes
point $\{c_{i}\}_{j=1}^{D}$ such that
\[
\varphi_{j}(c_{i})=\delta_{i,j} \quad \text{for } i,j\in[D].
\]
The benefit of such basis is that for any $f\in\text{span}(\varphi_{j})$,
we have that $f(t)=\sum_{j=1}^{D}f(c_{j})\varphi_{j}(t).$

For polynomials, the resulting basis are the Lagrange polynomials
\[
\varphi_{j}(t)=\prod_{ i \in [D] \backslash \{ j \} }\frac{t-c_{i}}{c_{j}-c_{i}} \quad \text{for } j\in[D].
\]
The only assumption we make on the basis is that its integral is bounded.
\begin{definition} \label{def:basis}
Given a $D$ dimensional subspace $\mathcal{V}\subset\mathcal{C}([0,T],\R)$
and node points $\{c_{j}\}_{j=1}^{D}\subset[0,T]$. For any $\gamma_{\varphi} \geq 1$, we call a basis
$\{\varphi_{j}\}_{j=1}^{D}\subset\mathcal{V}$ is $\gamma_{\varphi}$
bounded if $\varphi_{j}(c_{i})=\delta_{i,j}$ and 
we have
\[
\sum_{j=1}^{D}\left|\int_{0}^{t}\varphi_{j}(s) \d s \right|\leq \gamma_{\varphi}T
\quad \text{for } t\in[0,T].
\]
\end{definition}

Note that if the constant function $1\in\mathcal V$, then we have
$$1 = \sum_{j=1}^{D} 1(c_j) \varphi_j(t) = \sum_{j=1}^{D} \varphi_j(t).$$
Hence, we have
\[
T = \int^T_0 1 \d s \leq \sum_{j=1}^{D}\left|\int_{0}^{T}\varphi_{j}(s) \d s\right|\leq\gamma_{\varphi}T.
\]
Therefore, $\gamma_{\varphi} \geq 1$ for most of the interesting basis. This is the reason why we simply put it as an assumption to shorten some formulas.

In the section \ref{sec:piecewise_poly}, we prove that for the space of low degree polynomial
and piece-wise low degree polynomial, there is a basis on the Chebyshev
nodes that is $O(1)$ bounded.

Assuming that $\frac{dx}{dt}$ can be approximated by some element
in $\mathcal{V}$, we have that
\[
\frac{\d x}{\d t}(t)\sim\sum_{j=1}^{D}\frac{\d x}{\d t}(c_{j})\varphi_{j}(t)=\sum_{j=1}^{D}F(x(c_{j}),c_{j})\varphi_{j}(t).
\]
Integrating both sides, we have
\begin{equation}
x(t)\approx v+\int_{0}^{t}\sum_{j=1}^{D}F(x(c_{j}),c_{j})\varphi_{j}(s) \d s.\label{eq:approx_fix_point}
\end{equation}
This inspires us to consider the following operator from $\mathcal{C}([0,T],\R^{d})$
to $\mathcal{C}([0,T],\R^{d})$:
\begin{equation}
\mathcal{T}_{\varphi}(x)(t)_{i}=v_{i}+\int_{0}^{t}\sum_{j=1}^{D}F(x(c_{j}),c_{j})_{i}\varphi_{j}(s) \d s \quad \text{for } i\in[d]. \label{eq:def_T_d_operator}
\end{equation}
Equation (\ref{eq:approx_fix_point}) can be written as $x\approx\mathcal{T}_{\varphi}(x).$
To find $x$ to satisfies this, naturally, one can apply the fix point
iteration and this is called the collocation method.

\subsection{Collocation method}\label{sec:collocation}

From the definition of (\ref{eq:def_T_d_operator}), we note that $\mathcal{T}_{\varphi}(x)$
depends only on $x(t)$ at $t=c_{j}$. Therefore, we can only need
to calculate $\mathcal{T}_{\varphi}(x)(t)$ at $t=c_{j}$. To simplify
the notation, for any $x\in\mathcal{C}([0,T],\R^{d})$, we define
a corresponding matrix $[x]\in\R^{d\times D}$ by $[x]_{i,j}=x_{i}(c_{j})$.
For any $d\times D$ matrix $X$, we define $F(X,c)$ as an $d\times D$ matrix 
\begin{equation}
F(X,c)_{i,j}=F(X_{*,j},c_{j})_{i}.\label{eq:bF}
\end{equation}
where $X_{*,j}$ is the $j$-th column of $X$. Finally, we define $A_{\varphi}$ as a $D\times D$ matrix 
\begin{equation}
(A_{\varphi})_{i,j}=\int_{0}^{c_{j}}\varphi_{i}(s) \d s.\label{eq:Aphi}
\end{equation}
By inspecting the definition of (\ref{eq:def_T_d_operator}), (\ref{eq:bF})
and (\ref{eq:Aphi}), we have that
\[
[\mathcal{T}_{\varphi}(x)]=v \cdot 1_{D}^{\top}+F([x],c)A_{\varphi}
\]
where $1_{D}$ is a column of all 1 vector of length $D$. Hence, we
can apply the map $\mathcal{T}_{\varphi}$ by simply multiply $F([x],c)$
by a pre-compute $D\times D$ matrix $A_{\varphi}$. For the basis we considered in this paper, each iteration takes only $\tilde{O}(d D)$ which is nearly linear to the size of our representation of the solution. 




\begin{algorithm}\caption{Collocation Method}\label{alg:collocation_method}
\begin{algorithmic}[1]
\Procedure{\textsc{CollocationMethod}}{$F,v,T,\varphi,c$} \Comment{Theorem~\ref{thm:kth_order_ode}}
	\State Let $N =\left\lceil  \log\left( \frac{T}{\epsilon} \max_{s\in[0,T]}\left\Vert F(v,s)\right\Vert \right) \right\rceil $ \Comment{Choose number of iterations}
    \State Let $A_{\varphi}$ be the matrix defined by
$(A_{\varphi})_{i,j}=\int_{0}^{c_{j}}\varphi_{i}(s) \d s.$
	\State $X^{(0)}\leftarrow v \cdot 1_{D}^{\top}.$\Comment{$1_{D}$ is a column of all 1 vector of length $D$}
	\For{$j = 1, 2, \cdots, N-1$}
	    \State $X^{(j)}\leftarrow v \cdot 1_{D}^{\top}+F(X^{(j-1)},c)A_{\varphi}.$ \Comment{Matrix $F(X,c)$ is defined in Eq.~\eqref{eq:bF} }
	    \State \Comment{Note that we evaluate $D$ many $F$ every iteration in this matrix notation.}
	\EndFor
	\State $x^{(N)}(t) \leftarrow v + \int_0^t \sum_{i=1}^D F( X_{*,i}^{(N)} , c_i ) \varphi_i (s) \mathrm{d} s $
	\State \Return $x^{(N)}$
\EndProcedure
\end{algorithmic}
\end{algorithm}

We state our guarantee for a first-order ODE (Algorithm~\ref{alg:collocation_method}).

\begin{theorem}[First order ODE]\label{thm:first_order_ode}
Let $x^*(t) $ be the solution of an $d$ dimensional ODE 
\begin{align*}
x(0) = v , \frac{ \d x(t) }{ \d t } = F(x(t), t ) \quad \text{for all $0\leq t\leq T$}.
\end{align*}

We are given a $D$ dimensional subspace $\mathcal{V}\subset\mathcal{C}([0,T],\R)$, node points $\{c_{j}\}_{j=1}^{D}\subset[0,T]$ and a $\gamma_{\varphi}$ bounded basis
$\{\varphi_{j}\}_{j=1}^{D}\subset\mathcal{V}$ (Definition \ref{def:basis}).
Given some $L$ and $\epsilon>0$  such that

1. There exists a function $q\in\mathcal{V}$ such that
\begin{align*}
 \left\| q(t) - \frac{ \d }{ \d t } x^*(t) \right\| \leq \frac{ \epsilon }{ T } , \forall t \in [0,T].
\end{align*}

2. For any $y, z \in \R^{d}$, 
\begin{align*}
\| F(y, t ) - F(z , t) \| \leq L \|y -z \| , \forall t \in [0,T].
\end{align*}

Assume $\gamma_\varphi L T \leq 1/2$. Then the algorithm $\textsc{CollocationMethod}$ (Algorithm~\ref{alg:collocation_method}) outputs a function $x^{(N)}\in\mathcal{V}$ such that
\begin{align*}
\max_{t\in [0,T]}\| x^{(N)}(t) - x^* (t) \| \leq 20 \gamma_\varphi \epsilon .
\end{align*}
The algorithm takes $O \left( D \log\left( \frac{T}{\epsilon} \max_{s\in[0,T]}\left\Vert F(v,s)\right\Vert \right)   \right)$ evaluations of $F$.
\end{theorem}



Next we state the general result for a $k$-th order ODE.  We prove this via a reduction from higher order ODE to first-order ODE. See the proof in Appendix~\ref{sec:app_ode}.

\begin{theorem}[$k$-th order ODE]\label{thm:kth_order_ode}
Let $x^*(t) \in \R^d$ be the solution of the ODE 
\begin{align*}
\frac{ \d^k }{ \d t^{k} } x(t) = & ~ F \left( \frac{ \d^{k-1} }{ \d t^{k-1} } x(t) , \cdots , x(t), t \right) \notag\\
\frac{ \d^{i} }{ \d t^{i} } x(0) = & ~ v_{i}, \forall i \in \{ k-1, \cdots, 1 , 0 \}.
\end{align*}
where $F : \R^{k d + 1 } \rightarrow \R^d$, $x(t) \in \R^d$, and $v_0, v_1, \cdots, v_{k-1} \in \R^d$.

We are given a $D$ dimensional subspace $\mathcal{V}\subset\mathcal{C}([0,T],\R)$, node points $\{c_{j}\}_{j=1}^{D}\subset[0,T]$ and a $\gamma_{\varphi}$ bounded basis
$\{\varphi_{j}\}_{j=1}^{D}\subset\mathcal{V}$ (Definition \ref{def:basis}).
Given some $L$ and $\epsilon>0$  such that

1. For $i \in [k]$, there exists a function $q^{(i)}\in\mathcal{V}$ such that
\begin{align*}
\left\| q^{(i)}(t) - \frac{ \d^i }{ \d t^i } x^*(t) \right\| \leq \frac{ \epsilon }{ T^i }, \forall t \in [0,T].
\end{align*}

2. For any $y , z \in \R^{k d}$,
\begin{align*}
\| F( y, t ) - F( z , t) \| \leq \sum_{i=1}^k  L_i \| y_i( t ) - z_i( t ) \|, \forall t \in [0,T].
\end{align*}

Assume $\gamma_\varphi L T \leq 1/8$ with $L = \sum_{i=1}^k L_i^{1/i}$. Then, we can find functions $\{ q^{(i)} \}_{i\in\{0,1,\cdots,k-1\}}\subset \mathcal{V}$ such that
\begin{align*}
\max_{t \in [0,T] } \left\| q^{(i)}(t) - \frac{ \d^i }{ \d t^i } x^*(t) \right\|_p =  20 (1+2k) \gamma_\varphi \frac{\epsilon}{T^i}, \forall i \in \{0,1,\cdots, k-1\}.
\end{align*}
The algorithm takes $O( D \log(C/\epsilon)  )$ evaluations of $F$ where $$C = (4 \gamma_\varphi T)^{k}\cdot\max_{s\in[0,T]}\left\Vert F(v_{k-1},v_{k-2},\cdots,v_{0},s)\right\Vert +\sum_{i=1}^{k-1}(4 \gamma_\varphi T)^{i}\left\Vert v_{i}\right\Vert.$$
\end{theorem}

Note that the statement is a bit awkward. Instead of finding a function whose derivatives are same as the derivatives of $x^*$, the algorithm approximates the derivatives of $x^*$ individually. This is because we do not know if derivatives/integrals of functions in $\mathcal{V}$ remain in $\mathcal{V}$. For piece-wise polynomials, we can approximate the $j$-th derivative of the solution by taking $(k-j)$-th iterated integral of $q^{(k)}$, which is still a piece-wise polynomial. 
 
In section \ref{sec:piecewise_poly}, we give a basis for piece-wise polynomials (Lemma \ref{lem:basis_piecewise_poly}). Using this basis, we have the following Theorem.

\define{thm:kth_order_ode_piecewise}{Theorem}{{\rm($k$-th order ODE)}
Let $x^*(t) \in \R^d$ be the solution of the ODE 
\begin{align*}
\frac{ \d^k }{ \d t^{k} } x(t) = & ~ F \left( \frac{ \d^{k-1} }{ \d t^{k-1} } x(t) , \cdots , x(t), t \right) \notag\\
\frac{ \d^{i} }{ \d t^{i} } x(0) = & ~ v_{i}, \forall i \in \{ k-1, \cdots, 1 , 0 \}.
\end{align*}
where $F : \R^{k d + 1 } \rightarrow \R^d$, $x(t) \in \R^d$, and $v_0, v_1, \cdots, v_{k-1} \in \R^d$. Given some $L$ and $\epsilon>0$  such that

1. There exists a piece-wise polynomial $q(t)$ such that $q(t)$ on $[T_{j-1},T_j]$ is a degree $D_j$ polynomial with $$0 = T_0 < T_1 < \cdots < T_n = T$$ and that
\begin{align*}
\left\| q(t) - \frac{ \d^k }{ \d t^k } x^*(t) \right\| \leq \frac{ \epsilon }{ T^k }, \forall t \in [0,T]
\end{align*}

2. The algorithm knows about the intervals $[T_{j-1},T_j]$ and the degree $D_j$ for all $j \in [n]$.

3. For any $y , z \in \R^{k d}$,
\begin{align*}
\| F( y, t ) - F( z , t) \| \leq \sum_{i=1}^k  L_i \| y_i - z_i \|, \forall t \in [0,T].
\end{align*}
Assume $L T \leq 1/16000$ with $L = \sum_{i=1}^k L_i^{1/i}$. Then, we can find a piece-wise polynomial $x(t)$ such that
\begin{align*}
\max_{t \in [0,T] } \left\| \frac{ \d^i }{ \d t^i } x(t) - \frac{ \d^i }{ \d t^i } x^*(t) \right\|_p \lesssim   \frac{\epsilon k}{T^i} , \forall i \in \{0,1,\cdots, k-1\}.
\end{align*}
using $O( D \log(C/\epsilon)  )$ evaluations of $F$ with the size of basis $D=\sum_{ i = 1 }^n (1+D_i)$ and
$$O \left(d \min \left(\sum_{ i = 1 }^n (1+D_i)^2, D \log(C D/\epsilon) \right) \log(C / \epsilon) \right)$$ time where 
\begin{align*}
 C = O(T)^{k}\cdot \max_{s\in[0,T]} \left\Vert F(v_{k-1},v_{k-2},\cdots,v_{0},s)\right\Vert +\sum_{i=1}^{k-1} O(T)^{i}\left\Vert v_{i}\right\Vert .
\end{align*}
}

\state{thm:kth_order_ode_piecewise}
\begin{remark}
The two different runtime come from two different ways to the integrate of basis in Lemma \ref{lem:basis_piecewise_poly}. The first one is an navie method which is good enough for all our application. The second one follows from multipole method which gives an nearly linear time to the size of the basis with an extra log dependence on the accuracy.
\end{remark}
In the rest of this section, we prove the first-order guarantee, Theorem \ref{thm:first_order_ode}

\subsection{Proof of first order ODE}\label{sec:app_first_order_ode}

First, we bound the Lipschitz constant of the map $\T_\varphi$. Unlike the Picard-Lindel\"{o}f theorem, we are not able to get an improved bound of the Lipschitz constant of the composite of $\T_\varphi$. Fortunately, the Lipschitz constant of the map $\T_\varphi$ is good enough for all applications in this paper.

\begin{lemma}\label{lem:Lipschitz_Tphi}
Given any norm $\| \cdot \|$ on $\R^d$. Let $L$ be the Lipschitz constant of $F$ in $x$, i.e. 
$$\|F(x,s)-F(y,s)\| \leq L \|x-y\| \quad \text{for all } x, y\in \R^d, s\in[0,T].$$
Then, the Lipschitz constant of $\T_\varphi$ in this norm is upper bounded by $\gamma_\varphi LT$.
\end{lemma}
\begin{proof}
For any $0\leq t\leq T$,
\begin{align*}
\| \T_\varphi(x)(t) - \T_\varphi(y)(t) \|
= & ~ \left\| \int_0^t \sum_{j=1}^D F( x(c_j) , c_j ) \varphi_j(s) \d s - \int_0^t \sum_{j=1}^D F( y(c_j) , c_j ) \varphi_j(s) \d s \right\| \\
\leq & ~ \sum_{j=1}^D \left| \int_0^t \varphi_j(s) \d s \right| \cdot \max_{t \in [0,T]} \| F(x(t), t) - F(y(t), t) \| \\
\leq & ~ \gamma_{\varphi} L T \cdot \max_{s\in [0,t]} \| x(t) - y(t) \|\\
\leq & ~ \gamma_{\varphi} L T \| x - y \|.
\end{align*}
where the third step follows by $\sum_{j=1}^D | \int_0^t \varphi_j(s) \d s | \leq \gamma_{\varphi} T$ for all $0\leq t \leq T$.
\end{proof}

For the rest of the proof, let $x_\varphi^*$ denote the fixed point of $\T_\varphi$, i.e., $\T_\varphi( x_\varphi^* ) = x_\varphi^*$. The Banach fixed point theorem and Lemma \ref{lem:Lipschitz_Tphi} shows that $x_\varphi^*$ uniquely exists if $T \leq \frac{1}{L \gamma_\varphi}$.

Let $x^*$ denote the solution of the ODE, i.e., the fixed point of $\T$, with $\T(x^*) = x^*$.
Let $x^{(0)}$ denote the initial solution given by $x^{(0)}(t)=v$ and $x^{(N)}$ denote the solution obtained by applying operator $\T_\varphi$ for $N$ times. Note that $x^{(N)}(t)$ is the output of \textsc{CollocationMethod} in Algorithm~\ref{alg:collocation_method}.

Let $q\in\mathcal{V}$ denote an approximation of $\frac{ \d }{ \d t } x^*$ such that 
\begin{align*}
\left\|q(t) - \frac{ \d }{ \d t } x^*(t) \right\| \leq & ~ \frac{\epsilon}{T}. \forall t\in [0,T]
\end{align*}

The next lemma summarizes how these objects are related and will allow us to prove the main guarantee for first-order ODEs.

\begin{lemma}\label{lem:helper}
Let $L^{(j)}$ be the Lipschitz constant of the map $\T^{\circ j}_\varphi$. Assume that $L^{(N)} \leq 1/2$. Then, we have
\begin{eqnarray}
\| x^{(N)} - x^* \| &\leq& L^{(N)} \| x^{(0)} - x^* \| + 2 \| x_\varphi^* - x^* \|, \label{cla:ode_1_claim_1}\\
\| x_\varphi^* - x^* \| &\leq& 2 \cdot \|  \T_\varphi^{\circ N} (x^*) - x^* \|,  \label{cla:ode_1_claim_2} \\
\| x^* - \T_\varphi^{\circ N} (x^*) \| &\leq& \sum_{i=0}^{N-1} L^{(i)} \cdot \| x^* - \T_\varphi ( x^* ) \| , \label{cla:ode_1_claim_3}\\
\| x^* - \T_\varphi(x^*) \| &\leq& 2 \gamma_\varphi \cdot \epsilon. \label{cla:ode_1_claim_4}
\end{eqnarray}
\end{lemma}
\begin{proof}
We prove the claims in order.

For the first claim, 
\begin{align*}
\| x^{(N)} - x^* \| 
\leq & ~ \| x^{(N)} - x_\varphi^* \| + \| x_\varphi^* - x^* \| & \text{~by~triangle~inequality} \\
= & ~ \| \T_\varphi^{\circ N}( x^{(0)} ) - \T_\varphi^{\circ N}( x_\varphi^* ) \| + \| x_\varphi^* - x^* \|  \\
\leq & ~ L^{(N)} \| x^{(0)} - x_\varphi^* \| + \| x_\varphi^* - x^* \| \\
\leq & ~ L^{(N)} \| x^{(0)} - x^* \| + L^{(N)} \| x^* - x_\varphi^* \| + \| x_\varphi^* - x^* \|  \\
\leq & ~ L^{(N)} \| x^{(0)} - x^* \| + 2 \| x_\varphi^* - x^* \|
\end{align*}
where the last step follows by $L^{(N)} \leq 1$.

For the second claim, 
\begin{align*}
\| x_\varphi^* - x^* \| 
= & ~ \| \T_\varphi^{\circ N} ( x_\varphi^* ) -  x^*  \| & \text{~by~} x_\varphi^* = \T_\varphi^{\circ N} (x_\varphi^*) \\
\leq & ~ \| \T_\varphi^{\circ N } (x_\varphi^*) - \T_\varphi^{\circ N}(x^*) \| + \| \T_\varphi^{\circ N} (x^*) - x^* \| & \text{~by~triangle~inequality} \\
\leq & ~ L^{(N)} \cdot \| x_\varphi^* - x^* \| + \| \T_\varphi^{\circ N} (x^*) - x^* \| & \text{~by the definition of $L^{(N)}$} \\
\leq & ~ \frac{1}{2} \| x_\varphi^* - x^* \| + \| \T_\varphi^{\circ N} (x^*) - x^* \| & \text{~by~}L^{(N)} \leq 1/2
\end{align*}

For the third claim, 
\begin{align*}
\| x^* - \T_\varphi^{\circ N} (x^*) \| 
\leq & ~ \sum_{i=0}^{N-1} \| \T_\varphi^{\circ i} (x^*) - \T_\varphi^{\circ (i+1)} ( x^* ) \| \\
\leq & ~ \sum_{i=0}^{N-1} L^{(i)} \cdot \| x^* - \T_\varphi ( x^* ) \| \\
\end{align*}
For the last claim, 
\begin{align*}
 & ~ \| x^* (t) - \T_\varphi(x^*) (t) \| \\
= & ~ \| \T(x^*) (t) - \T_\varphi(x^*) (t)\| \\ 
= & ~ \left\| \int_0^t  F( x^*(s), s ) \d s - \int_0^t \sum_{j=1}^D F(x^*(c_j) , c_j) \varphi_j(s) \d s \right\| \\ 
= & ~ \left\| \int_0^t \frac{ \d }{ \d t } x^*(s) \d s - \int_0^t \sum_{j=1}^D \frac{ \d }{ \d t} x^*(c_j) \varphi_j(s) \d s \right\| \\ 
\leq & ~ \left\| \int_0^t (\frac{ \d }{ \d t } x^*(s) - q(s) ) \d s - \int_0^t \sum_{j=1}^D (\frac{ \d }{ \d t}  x^*(c_j) - q(c_j) ) \varphi_j(s) \d s \right\| \\
+ & ~ \left\| \int_0^t q(s) \d s - \int_0^t \sum_{j=1}^D q(c_j) \varphi_j(s) \d s \right\| \\ 
\leq & ~ \int_{0}^{t}\left\Vert \frac{\d}{\d t} x^{*}(s)-q(s)\right\Vert \d s+\sum_{j=1}^{D}\left\Vert \frac{\d}{\d t}x^{*}(c_{j})-q(c_{j})\right\Vert \left|\int_{0}^{t}\varphi_{j}(s)\d s\right| + 0 \\
\leq & ~ (1+\gamma_\varphi) \cdot \epsilon + 0
\end{align*}
where the first step follows by ${\cal T}(x^*) = x^*$, the second step follows by the definition of ${\cal T}$ and ${\cal T}_{\varphi}$, the third step follows by $x^*(t)$ is the solution of ODE, the fourth step follows by triangle inequality, the second last step follows by $q \in\mathcal{V}$, and the last step follows by $\| \frac{\d }{ \d t}  x^* - q  \| \leq \frac{ \epsilon }{ T }$ and the definition of $\gamma_\varphi$.
\end{proof}

Now, we are ready to prove Theorem~\ref{thm:first_order_ode}.
\begin{proof}
Using Lemma \ref{lem:helper}, we have
\begin{align*}
\| x^{(N)} - x^* \| 
\leq & ~ L^{(N)} \| x^{(0)} - x^* \| + 2 \| x_\varphi^* - x^* \| & \text{~by~Eq.~\eqref{cla:ode_1_claim_1}} \\
\leq & ~ L^{(N)} \| x^{(0)} - x^* \| + 4 \| \T_\varphi^{\circ N} (x^*) - x^* \| &\text{~by~Eq.~\eqref{cla:ode_1_claim_2}} \\
\leq & ~ L^{(N)} \| x^{(0)} - x^* \| + 4 \sum_{i=0}^{N-1} L^{(i)} \cdot \| x^* - \T_\varphi(x^*) \| &\text{~by~Eq.~\eqref{cla:ode_1_claim_3}} \\
\leq & ~ L^{(N)}  \| x^{(0)} - x^*\| + 8 \sum_{i=0}^{N-1} L^{(i)} \cdot \gamma_\varphi \cdot \epsilon. & \text{~by~Eq.~\eqref{cla:ode_1_claim_4}}
\end{align*}

Using the assumption that $\gamma_{\varphi}LT\leq\frac{1}{2}$, Lemma
\ref{lem:Lipschitz_Tphi} shows that $L^{(1)}\leq\frac{1}{2}$ and hence $L^{(j)}\leq\frac{1}{2^{j}}$.
Therefore, we have
\begin{equation}\label{eq:first_order_error}
\|x^{(N)}-x^{*}\|\leq \frac{1}{2^{N}} \|x^{(0)}-x^{*}\| +16 \gamma_{\varphi}\cdot\epsilon =  \frac{1}{2^{N}} \|x^{*}-x^{*}(0)\| +16 \gamma_{\varphi}\cdot\epsilon
\end{equation}

To bound $\|x^{*}-x^{*}(0)\|$, for any $0\leq t\leq T$
\[
x^{*}(t)=x^{*}(0)+\int_{0}^{t}F(x^{*}(s),s) \d s.
\]
Hence, we have that
\begin{align*}
\|x^{*}(t)-x^{*}(0)\| & \leq\left\Vert \int_{0}^{T}F(x^{*}(0),s) \d s \right\Vert +\left\Vert \int_{0}^{t}\left(F(x^{*}(s),s)-F(x^{*}(0),s)\right) \d s \right\Vert \\
 & \leq\left\Vert \int_{0}^{T}F(x^{*}(0),s) \d s \right\Vert +L\int_{0}^{t}\|x^{*}(s)-x^{*}(0)\| \d s.
\end{align*}
Solving this integral inequality (see Lemma~\ref{lem:exp_ODE}), we have that
\[
\|x^{*}(t)-x^{*}(0)\|\leq e^{Lt}\left\Vert \int_{0}^{T}F(x^{*}(0),s) \d s\right\Vert .
\]
Now, we use $LT\leq\frac{1}{2}$ and get
\[
\|x^{*}(t)-x^{*}(0)\|\leq2\left\Vert \int_{0}^{T}F(x^{*}(0),s) \d s\right\Vert .
\]

Picking $N = \left\lceil \log_{2}\left( \frac{T}{\epsilon} \max_{s\in[0,T]} \left\Vert F(x^{*}(0),s) \right\Vert \right)\right\rceil $, (\ref{eq:first_order_error}) shows that the error is less than $20 \gamma_{\varphi}\epsilon$.

\end{proof}

\subsection{A basis for piece-wise polynomials}\label{sec:piecewise_poly}

In this section, we discuss how to construct a bounded basis for
low-degree piece-wise polynomials. We are given $n$ intervals $\{[I_{i-1},I_{i}]\}_{i=1}^{n}$
where $I_{0}=0$ and $I_{n}=T$. In the $i^{th}$ interval $[I_{i-1},I_{i}]$,
we represent the function by a degree $D_{i}$ polynomial. Formally,
we define the function subspace by
\begin{align}
\mathcal{V} & \defeq\bigoplus_{i=1}^{n}\mathcal{V}_{i} \quad \text{with} \quad \mathcal{V}_{i} \defeq \left\{ \left(\sum_{j=0}^{D_{i}}\alpha_{j}t^{j} \right)\cdot1_{[I_{i-1},I_{i}]}:\alpha_{j}\in\R \right\}.\label{eq:piece_wise_poly}
\end{align}
The following Lemma shows we can construct the basis for $\mathcal{V}$
by concatenating the basis for $\mathcal{V}_{i}$.
\begin{lemma}
\label{lem:concat_basis}For $i\in[n]$, we are given a $\gamma_{i}$
bounded basis $\{\varphi_{j,i}\}_{j=0}^{D_{i}}$ for the subspace
$\mathcal{V}_{i}\subset\mathcal{C}([I_{i-1},I_{i}],\R)$ on nodes
point $\{c_{j,i}\}_{j=0}^{D_{i}}$. Then, $\{\varphi_{j,i}\}_{i,j}$
is a $\sum_{i=1}^n \gamma_{i}(I_{i}-I_{i-1})$ bounded basis for the subspace
$\bigoplus_{i=1}^{n}\mathcal{V}_{i} \subset\mathcal{C}([I_{0},I_{n}],\R)$.
\end{lemma}

\begin{proof}
For any $t\geq 0$, we have
\begin{align*}
\sum_{ i = 1 }^n \sum_{ j = 0 }^{ D_i } \left|\int_{I_{0}}^{t}\varphi_{i,j}(s)\d s \right| & \leq \sum_{ i = 1 }^n \left(\sum_{ j = 0 }^{ D_i } \left|\int_{I_{i-1}}^{t}\varphi_{i,j}(s) \d s\right|1_{t\geq I_{i-1}}\right)\\
 & \leq\sum_{ i = 1 }^n \gamma_{i}(I_{i}-I_{i-1})
\end{align*}
where we used that $\varphi_{i,j}$ is supported on $[I_{i-1},I_{i}]$
in the first inequality.
\end{proof}
Next, we note that the boundedness for basis is shift and scale invariant.
Hence, we will focus on obtaining a basis for $(t-1)$-degree polynomial
on $[-1,1]$ for notation convenience.

For $[-1,1]$, we choose the node points $c_{j}=\cos(\frac{2j-1}{2t}\pi)$
and the basis are
\[
\varphi_{j}(x)=\frac{\sqrt{1-c_{j}^{2}}\cos(t\cos^{-1}x)}{t(x-c_{j})}.
\]
It is easy to see that $\varphi_{j}(c_{i})=\delta_{i,j}$. To bound
the integral, Lemma 91 in \cite{lv17a} shows that
\[
\left|\int_{-1}^{y}\varphi_{j}(x)dx\right|\leq\frac{2000}{t}\text{ for all }y\in[-1,1].
\]
Summing it over $t$ basis functions, we have that $\gamma_{\varphi}\leq2000$.
Together with Lemma \ref{lem:concat_basis}, we have the following
result:
\begin{lemma}
\label{lem:basis_piecewise_poly}Let $\mathcal{V}$ be a subspace
of piecewise polynomials on $[0,T]$ with fixed nodes. Then, there is a $2000$ bounded
basis $\{\varphi\}$ for $\mathcal{V}$. Furthermore, for any vector $v$, it takes $O(\sum_{i=1}^n (1+D_i)^2)$ time to compute $v^\top A_\varphi$ where $D_i$ is the maximum degree of the $i$-th piece.

Alternatively, one can find $u$ such that $\|u - v^\top A_\varphi\| \leq \epsilon T \|v\|_\infty$ in time
$$O\left(\mathrm{rank}(\mathcal{V}) \log(\frac{\mathrm{rank}(\mathcal{V})}{\epsilon}) \right).$$
\end{lemma}
\begin{proof}
The bound follows from previous discussion. For the computation cost, note that
\[
(v^{\top}A_{\varphi})_{(i,j)}=\sum_{i',j'}\int_{0}^{c_{(i,j)}}v_{(i',j')}\varphi_{(i',j')}(s) \d s.
\]
where $(i,j)$ is the $j$-th node at the $i$-th piece.
For any $i\neq i'$, the support of $\varphi_{(i',j')}(s)$ is either
disjoint from $[0,c_{(i,j)}]$ (if $i'>i$) or included in $[0,c_{(i,j)}]$
(if $i'<i$). Hence, we have that
\[
\int_{0}^{c_{(i,j)}}\varphi_{(i',j')}(s) \d s=\begin{cases}
0 & \text{if }i'>i\\
\int_{-\infty}^{\infty}\varphi_{(i',j')}(s) \d s & \text{if }i'<i
\end{cases}.
\]
Therefore, we have
\[
(v^{\top}A_{\varphi})_{(i,j)}=\sum_{i'<i,j'}v_{(i',j')}\cdot\int_{-\infty}^{\infty}\varphi_{(i',j')}(s) \d s+\sum_{j'}v_{(i,j')}\cdot\int_{0}^{c_{(i,j)}}\varphi_{(i,j')}(s) \d s.
\]
Note that $\int_{0}^{c_{(i,j)}}\varphi_{(i',j')}(s) \d s$ can precomputed.
Since there are $\sum_{i=1}^n (1+D_{i})$ many pairs of $(i,j)$, the first
term can be computed in $\sum_{i=1}^n (1+D_{i})$ time. Since there are
$\sum_{i=1}^n (1+D_{i})^{2}$ many pairs of $(i,j,j')$, the second term
can be computed in $\sum_{i=1}^n (1+D_{i})^{2}$ time.

Theorem \ref{thm:multipole} gives another way to compute the integration and its runtime is
$$O\left(\mathrm{rank}(\mathcal{V}) \log(\frac{\mathrm{rank}(\mathcal{V})}{\epsilon}) \right).$$
\end{proof}
\begin{remark}
Experiment seems to suggest the basis we proposed is $1$ bounded.
\end{remark}

Here is the theorem we used above to compute the Lagrange polynomials.
\begin{theorem}[{\cite[Section 5]{dutt1996fast}}]\label{thm:multipole}
Let $\phi_{i}$ be the Lagrange basis polynomials on the Chebyshev
nodes $c_{j}=\cos(\frac{2j-1}{2t}\pi)$ for $j\in[t]$, namely, $\phi_{i}(s)=\prod_{j\neq i}\frac{s-c_{j}}{c_{i}-c_{j}}$.
Given a polynomial $p(s)=\sum_{j=1}^{t}\alpha_{j}\phi_{j}(s)$ represented
by $\{\alpha_{j}\}_{j=1}^{t}$, one can compute $\{\ell_{i}\}_{i=1}^{t}$
such that
\[
\left|\ell_{i}-\int_{0}^{c_{i}}p(s)ds\right|\leq\epsilon \|\alpha\|_\infty
\]
in time $O(t\log(\frac{t}{\epsilon}))$.
\end{theorem}

\section{Improved Contraction Bound for HMC}\label{sec:contraction}

In this section, we give an improved contraction bound for HMC (Algorithm~\ref{alg:hmc}). Each iteration of Algorithm~\ref{alg:hmc} solve the HMC dynamics approximately.
In the later sections, we will discuss how to solve this ODE.

To give a contraction bound for the noisy HMC, we first analyze the
contraction of the ideal HMC. We reduce the problem of bounding the
contraction rate of the ideal HMC to the following lemma involving a matrix ODE.

\begin{lemma}\label{lem:matrix_ODE_bound}
Given a symmetric matrix $H(t)$ such
that $0 \prec m_2 \cdot I\preceq H(t)\preceq M_2 \cdot I$ for all $t\geq0$.
Consider the ODE
\begin{align*}
u''(t) & =-H(t)\cdot u(t),\\
u'(0) & =0.
\end{align*}
Let $\alpha(t)=\frac{1}{\|u(0)\|_{2}}\int_{0}^{t}(t-s)\cdot\|H(s)u(0)\|_{2} \d s$.
For any $0\leq T\leq\frac{1}{2\sqrt{M_2}}$ such that $\alpha(T)\leq\frac{1}{8}\sqrt{\frac{m_2}{M_2}}$,
we have that
\[
\|u(T)\|_{2}^{2} \leq \left( 1-\max\left( \frac{1}{4} m_2 T^{2}, \frac{1}{2} \sqrt{\frac{ m_2 }{ M_2 }}\cdot\alpha(T)\right) \right) \cdot \|u(0)\|_{2}^{2}.
\]
\end{lemma}

Using this lemma, we prove both parts of the main contraction bound, Lemma \ref{lem:HMC_contraction}.

\begin{proof}[Proof of Lemma \ref{lem:HMC_contraction}.]
Let error function $e(t)=y(t)-x(t)$. The definition of HMC shows that
\[
e''(t)=-(\nabla f(y(t))-\nabla f(x(t)))=-H(t)\cdot e(t)
\]
where $H(t)=\int_{0}^{1}\nabla^{2}f(x(t)+s(y(t)-x(t))) \d s$. By the
strong convexity and the Lipschitz gradient of $f$, we have that
\[
m_2 \cdot I\preceq H(t)\preceq M_2 \cdot I.
\]
Hence, we can apply Lemma \ref{lem:matrix_ODE_bound}. 

To get the first bound, we bound the $\alpha(t)$ defined in Lemma
\ref{lem:matrix_ODE_bound} as follows
\[
\alpha(t)=\frac{1}{\|e(0)\|_{2}}\int_{0}^{t}(t-s)\cdot\|H(s)e(0)\|_{2} \d s\leq M_2 \int_{0}^{t}(t-s) \d s = M_2 \frac{t^{2}}{2}.
\]
Therefore, for $0\leq t\leq\frac{ m_2^{1/4} }{2 M_2^{3/4} }$, we have that
$\alpha(t)\leq\frac{1}{8}\sqrt{ \frac{ m_2 }{ M_2 } }$ and hence Lemma \ref{lem:matrix_ODE_bound}
gives the first bound.

To get the second bound, we note that $\alpha(t)$ is increasing and
hence there is $t$ such that $\alpha(t)=\frac{1}{8}\sqrt{\frac{ m_2 }{ M_2 }}$.
Using such $t$ in Lemma \ref{lem:matrix_ODE_bound} gives the second
bound.
\end{proof}
Now, we prove the main technical lemma of this section, a contraction
estimate for matrix ODE. We note that not all matrix ODEs come from
some HMC and hence it might be possible to get a better bound by directly
analyzing the HMC.

\begin{proof}[Proof of Lemma \ref{lem:matrix_ODE_bound}]
Let $e_1$ denote the basis vector that it is $1$ in the first coordinate and $0$ everywhere else.

Without loss of generality, we can assume $\|u(0)\|_{2}=1$ and $u(0)=e_{1}$. 

The proof involves first getting a crude bound on $\|u(t)\|_{2}$.
Then, we boast the bound by splitting the movement of $u(t)$ into
one parallel to $e_1$ and one orthogonal to $e_{1}$.

\textbf{Crude bound on $\|u(t)\|_{2}$:}

Integrating both sides of $u''(t)=-H(t)\cdot u(t)$ twices and using $u'(0)=0$ gives
\begin{equation}
u(t)=u(0)-\int_{0}^{t}(t-s)H(s)u(s) \d s.\label{eq:crude_u_bound}
\end{equation}
We take the norm on both sides and use $0\preceq H(s)\preceq M_2 \cdot I$
to get
\[
\|u(t)\|_{2}\leq 1 + M_2 \cdot \int_{0}^{t}(t-s)\|u(s)\|_{2} \d s.
\]
Applying Lemma \ref{lem:cosh_ODE} to this equation and using $t\leq\frac{1}{2\sqrt{M_2}}$
gives
\[
\|u(t)\|_{2}\leq\cosh( \sqrt{M_2} t )\leq\frac{6}{5}
\]
Putting it back to (\ref{eq:crude_u_bound}) gives
\[
\|u(t)-e_1\|_{2}\leq\int_{0}^{t}(t-s)\|H(s)\cdot u(s)\|_{2} \d s=\int_{0}^{t}(t-s)\cdot M_2 \cdot \frac{6}{5} \d s=\frac{6}{10} M_2 t^{2}.
\]
In particular, for any $0\leq t\leq\frac{1}{2\sqrt{ M_2 }}$, we have
that
\begin{equation}
\frac{5}{6}\leq u_{1}(t)\leq\frac{7}{6}.\label{eq:u1_bound}
\end{equation}

\textbf{Improved bound on $\|u(t)\|_{2}$:}

Let $P_{1}$ be the orthogonal projection to the first coordinate
and $P_{-1}=I-P_{1}$. We write $u(t)=u_{1}(t)+u_{-1}(t)$ with $u_{1}(t)=P_{1}u(t)$
and $u_{-1}(t)=P_{-1}u(t)$, namely, $u_{1}(t)$ is parallel to $e_1$
and $u_{-1}(t)$ is orthogonal to $e_1$. 

Fix any $0\leq t\leq T$. Let $\beta(t)=e_1^{\top}u(t)$. By the
definition of $u$, we have
\begin{equation}
u''(t)=-\beta(t)\cdot H(t)e_1-H(t)u_{-1}(t).\label{eq:ddu}
\end{equation}
Integrating both sides twice and using $u_{-1}(0)=0$, we have
\begin{align*}
u_{-1}(t) & =\int_{0}^{t}(t-s)P_{-1}u''(s) \d s\\
 & =-\int_{0}^{t}(t-s)\cdot\beta(s)\cdot P_{-1}H(s)e_1 \d s-\int_{0}^{t}(t-s)P_{-1}H(s)u_{-1}(s) \d s.
\end{align*}
Taking norm on both sides and using that $0\preceq H(t)\preceq M_2 \cdot I$
and $\frac{5}{6}\leq\beta(s)\leq\frac{7}{6}$, we have that
\begin{align}
\|u_{-1}(t)\|_{2} 
\leq & ~ \int_{0}^{t}(t-s)\cdot\beta(s)\cdot\|H(s)e_1\|_{2} \d s+M_2 \cdot\int_{0}^{t}(t-s)\cdot\|u_{-1}(s)\|_{2} \d s\nonumber \\
\leq & ~ \frac{7}{6} \int_{0}^{t}(t-s)\cdot \|H(s)e_1\|_{2} \d s+M_2 \cdot\int_{0}^{t}(t-s)\cdot\|u_{-1}(s)\|_{2} \d s\nonumber \\
\leq & ~ \frac{7}{6}\alpha(T)+M_2 \cdot\int_{0}^{t}(t-s)\cdot\|u_{-1}(s)\|_{2} \d s\label{eq:u-1_bound}
\end{align}
where $\alpha(T)\defeq\int_{0}^{T}(T-s)\cdot\|H(s)e_1\|_{2} \d s.$
Solving this integral inequality (Lemma \ref{lem:cosh_ODE}), we get
\begin{equation}
\|u_{-1}(t)\|_{2} \leq \frac{7}{6}\alpha(T)\cdot\cosh(\sqrt{M_2}t) \leq \frac{7}{6} \alpha(T) \cdot \cosh ( 1 /2 ) \leq \frac{4}{3}\alpha(T)\label{eq:u-1_bound2}
\end{equation}
where the second step follows from $t \leq \frac{1}{2\sqrt{M_2}}$, and the last step follows from $\cosh(1/2) \leq \frac{8}{7}$ .

Next, we look at the first coordinate of (\ref{eq:ddu}) and get
\begin{align}
\beta''(t) & =-\beta(t)\cdot e_1^{\top}H(t)e_1-e_1^{\top}H(t)u_{-1}(t) \notag \\
 & \leq-\frac{5}{6}e_1^{\top}H(t)e_1+\|H(t)e_1\|_{2}\cdot\|u_{-1}(t)\|_{2}. \label{eq:ideal_HMC_beta}
\end{align}
To bound the last term, we note that
\begin{align}\label{eq:Hs_bound}
\|H(t)e_1\|_{2} 
= & ~ \sqrt{e_1^{\top}H^{2}(t)e_1} \notag \\
\leq & ~ \sqrt{M_2 \cdot e_1^{\top}H(t)e_1} & \text{~by~} 
H^2(t) \preceq M_2 \cdot H(t) \notag \\
= & ~ \sqrt{ \frac{M_2}{m_2} } \sqrt{\frac{m_2}{ e_1^\top H(t) e_1 }} e_1^\top H(t) e_1 \notag \\
\leq & ~ \sqrt{\frac{ M_2 }{ m_2 }}\cdot e_1^{\top} H(t)e_1 & \text{~by~} m_2 \leq e_1^\top H(t) e_1.
\end{align}

Using this into (\ref{eq:ideal_HMC_beta}) and $\alpha(T)\leq\frac{1}{8}\sqrt{\frac{m_2}{M_2}}$, we have
\[
\beta''(t)\leq-e_1^{\top}H(t)e_1\cdot\left(\frac{5}{6}-\frac{4}{3}\sqrt{\frac{ M_2 }{ m_2 }}\alpha(T)\right)\leq-\frac{2}{3}e_1^{\top}H(t)e_1.
\]
Hence, we have that
\begin{equation}
\beta(t)\leq1-\frac{2}{3}\int_{0}^{t}(t-s)\cdot e_1^{\top}H(s)e_1 \d s.\label{eq:alpha_bound}
\end{equation}

Using (\ref{eq:alpha_bound}), $\beta(t)\geq\frac{5}{6}$ and (\ref{eq:u-1_bound2})
gives
\begin{align*}
\|u(t)\|_{2}^{2} & =\beta^{2}(t)+\|u_{-1}(t)\|_{2}^{2}\\
 & \leq1-\int_{0}^{t}(t-s)\cdot e_1^{\top}H(s)e_1 \d s+\left(\frac{4}{3}\int_{0}^{T}(T-s)\cdot\|H(s)e_1\|_{2} \d s\right)^{2}\\
 & \leq1-\int_{0}^{t}(t-s)\cdot e_1^{\top}H(s)e_1 \d s+2\sqrt{\frac{M_2}{m_2}}\int_{0}^{T}(T-s)\cdot e_1^{\top}H(s)e_1 \d s\cdot\alpha(T)\\
 & =1-\int_{0}^{t}(t-s)\cdot e_1^{\top}H(s)e_1 \d s\left(1-2\sqrt{\frac{M_2}{m_2}}\alpha(T)\right)\\
 & \leq1-\frac{1}{2}\int_{0}^{t}(t-s)\cdot e_1^{\top}H(s)e_1 \d s
\end{align*}
where we used (\ref{eq:Hs_bound}) at the second inequality and $\alpha(T)\leq\frac{1}{8}\sqrt{\frac{m_2}{M_2}}$
at the end.

Finally, we bound the last term in two way. One way simply uses $e_1^{\top}H(s)e_1 \geq m_2$ 
and get
\begin{align*}
\int_0^t ( t - s ) \cdot e_1^\top H(s) e_1 \d s \geq \frac{m_2}{2} t^2
\end{align*}
which implies 
\begin{align*}
\|u(t)\|_{2}^{2}\leq 1 - \frac{m_2}{4} t^{2}.
\end{align*}

For the other way, we apply (\ref{eq:Hs_bound}) to get that
\begin{align*}
\int_{0}^{t}(t-s)\cdot e_1^{\top}H(s)e_1 \d s
\geq \int_{0}^{t}(t-s)\cdot \| H(t) e_1 \|_2 \sqrt{ \frac{ m_2 }{ M_2 } } \d s 
= \sqrt{\frac{m_2}{M_2}}\cdot\alpha(T).
\end{align*}
where the last step follows by the definition of $\alpha(T)$. Thus, we have
\begin{align*}
\| u(t) \|_2^2 \leq 1 - \frac{1}{2} \sqrt{ \frac{ m_2 }{ M_2 } } \alpha(T).
\end{align*} 

\end{proof}

Finally, we analyze the contraction of the noisy HMC.
\begin{theorem}[Contraction of noisy HMC]\label{thm:contraction_HMC}
Suppose $f$ is $m_2$ strongly convexity
with $M_2$ Lipschitz gradient. For any step-size $h \leq \frac{ m_2^{1/4} }{2 M_2^{3/4}} $, let $X\sim \textsc{HMC}(x^{(0)},f,\epsilon,h)$ and $Y\sim e^{-f}$. Then, we have that
\[
W_{2}(X,Y)\leq \frac{\epsilon}{\sqrt{m_2}}.
\]
In addition, the number of iterations is
\begin{align*}
N = O \left(\frac{1}{m_2 h^2} \right) \cdot \log \left(\frac{\|\nabla f(x^{(0)})\|_{2}^{2} / m_2 + d}{\epsilon} \right).
\end{align*}
\end{theorem}

\begin{proof}
To prove the $W_{2}$ distance, we let $x^{(k)}$ be the iterates
of the algorithm $\textsc{HMC}$. Let $y^{(k)}$ be the $k$-th step of the
ideal HMC starting from a random point $y^{(0)}\sim e^{-f}$ with
the random initial direction identical to the algorithm $\textsc{HMC}$.
Let $x^{*(k)}$ be the 1 step ideal HMC starting from $x^{(k-1)}$ with the same initial direction as $y^{(k)}$. Lemma \ref{lem:HMC_contraction}
shows that
\[
\|x^{*(k)}-y^{(k)}\|_{2}^{2}\leq\left(1-\frac{m_2 h^2}{4}\right)\|x^{(k-1)}-y^{(k-1)}\|_{2}^{2}.
\]
Let $\theta = \frac{ m_2 h^2 }{8}$, then $\theta \leq \frac{1/(4 \kappa^{3/4})}{8} \leq 1/32$. By the assumption of the
noise, we have $\|x^{(k)}-x^{*(k)}\|_2 \leq \frac{\epsilon \theta}{2  \sqrt{m_2}}$. Hence,
we have
\begin{align*}
\|x^{(k)}-y^{(k)}\|_{2}^{2} 
= & ~ \| ( x^{*(k)} - y^{(k)} ) + ( x^{(k)} - x^{*(k)} ) \|_2^2 \\
\leq & ~ \left( 1 + \theta \right) \|x^{*(k)}-y^{(k)}\|_{2}^{2}+(1 + 1/\theta )\|x^{(k)}-x^{*(k)}\|_2^{2}\\
\leq & ~ \left( 1 +  \theta \right) ( 1 - 2  \theta ) \|x^{(k-1)}-y^{(k-1)}\|_{2}^{2}+(1+ 1/ \theta )\|x^{(k)}-x^{*(k)}\|_2^{2} \\
\leq & ~ \left( 1 - \theta \right)\|x^{(k-1)}-y^{(k-1)}\|_{2}^{2}+ (1 +  1/\theta )\|x^{(k)}-x^{*(k)}\|_2^{2} \\
\leq & ~ \left( 1 - \theta \right)\|x^{(k-1)}-y^{(k-1)}\|_{2}^{2}+ (2 /\theta)  \cdot \frac{\epsilon^2 \theta^2}{4 m_2}.
\end{align*}
where the second step follows by $(a+b)^2 \leq (1+ \theta) a^2 + (1+ 1/\theta)b^2$, the third step follows by $ \|x^{*(k)}-y^{(k)}\|_{2}^{2} \leq (1-2\theta) \|x^{(k-1)}-y^{(k-1)}\|_{2}^{2} $, the fourth step follows by $(1+\theta)(1-2\theta) \leq (1-\theta)$, the fifth step follow by $\theta \leq 1/4$ and $\| x^{(k)} - x^{*(k)} \|_2 \leq \frac{\epsilon \theta}{2 \sqrt{m_2}}$. 

Applying this bound iteratively gives 
\begin{equation}
\|x^{(k)}-y^{(k)}\|_{2}^{2} \leq \left( 1- \theta \right)^{k}\|x^{(0)}-y^{(0)}\|_{2}^{2}+\frac{\epsilon^{2}}{2 m_2}.\label{eq:xyk_bound}
\end{equation}
Let $x^{(\min)}$ be the minimum of $f$. Then, we have
\begin{align}
\|x^{(0)}-y^{(0)}\|_{2}^{2} & \leq 2\|x^{(0)}-x^{(\min)}\|_2^{2}+2\|y^{(0)}-x^{(\min)}\|_2^{2}.\label{eq:initial_bound}
\end{align}
For the first term, the strong convexity of $f$ shows that
\begin{equation}
\|x^{(0)}-x^{(\min)}\|_2^{2}\leq\frac{1}{m^{2}}\|\nabla f(x^{(0)})\|_{2}^{2}.\label{eq:initial_first_term}
\end{equation}
For the second term, Theorem 1 in \cite{dm16} 
shows that
\begin{equation}
\E \left[ \|y^{(0)}-x^{(\min)}\|_2^{2} \right] \leq\frac{d}{m_2}.\label{eq:initial_second_term}
\end{equation}

Combining (\ref{eq:xyk_bound}), (\ref{eq:initial_bound}), (\ref{eq:initial_first_term})
and (\ref{eq:initial_second_term}), we have
\[
\E \left[ \|x^{(k)}-y^{(k)}\|_{2}^{2} \right] \leq \left(1- \theta \right)^{k}\left(\frac{2\|\nabla f(x^{(0)})\|_{2}^{2}}{ m_2^{2} }+\frac{2d}{ m_2 }\right)+\frac{\epsilon^{2}}{2 m_2}.
\]
Picking 
\begin{align*}
k = \frac{1}{\theta} \cdot\log\left(\frac{\frac{2\|\nabla f(x^{(0)})\|_{2}^{2}}{ m_2^{2} }+\frac{2d}{ m_2 }}{\frac{\epsilon^{2}}{2 m_2}}\right) = \frac{1}{\theta} \cdot\log\left(\frac{4}{\epsilon^2} \left(\frac{\|\nabla f(x^{(0)})\|_{2}^{2}}{ m_2}+d \right)\right) ,
\end{align*}
we have that
\[
\E \left[ \|x^{(k)}-y^{(k)}\|_{2}^{2} \right] \leq \frac{\epsilon^{2}}{m_2}.
\]
This proves that $W_{2}(X,Y) \leq \frac{\epsilon}{\sqrt{m_2}}$.
\end{proof}

\section{Strongly Convex functions with Lipschitz Gradient}
\label{sec:nonsmooth}

In this section, we give a faster sampling algorithm for strongly convex functions
with Lipschitz gradient. The purpose of this section is to illustrate
that our contraction bound and ODE theorem are useful even for functions that are not infinitely differentiable. We believe that our bound can
be beaten by algorithms designed for this specific setting.

\subsection{Bounding the ODE solution}
First, we prove that the HMC dynamic for these functions can be well
approximated by piece-wise degree-2 polynomials. Note that this only requires that the Hessian has bounded eigenvalues. 
\begin{lemma}[Smoothness implies the existence of degree-2 polynomial approximation]
\label{lem:2nd_poly} Let $f$ be a twice-differentiable function such that $-M_{2}\cdot I\preceq\nabla^{2}f(x)\preceq M_{2}\cdot I$
for all $x\in\mathbb{R}^{d}$. Let $0\leq h\leq\frac{1}{2\sqrt{M_{2}}}$.
Consider the HMC dynamic 
\begin{align*}
\frac{\d^{2}x}{\d t^{2}}(t) & =-\nabla f(x(t))\text{ for }0\leq t\leq h,\\
\frac{\d x}{\d t}(0) & =v_{1},\\
x(0) & =v_{0}.
\end{align*}
For any integer $k$, there is a continuously differentiable $k$-piece
degree 2 polynomial $q$ such that $q(0)=v_{0}$, $\frac{\d q}{\d t}(0)=v_{1}$
and $\left\| \frac{\d^{2}q}{\d t^{2}}(t)-\frac{\d^{2}x}{\d t^{2}}(t)\right\| _{2}\leq\frac{\epsilon}{h^{2}}\text{ for }0\leq t\leq h$
with 
\[
\epsilon=\frac{2M_{2}h^{3}}{k}\left(\|v_{1}\|_{2}+\|\nabla f(v_{0})\|_{2}\cdot h\right).
\]
\end{lemma}

\begin{proof}
Let $x(t)$ be the solution of the ODE. We define a continuously differentiable
$k$-piece degree-2 polynomial $q$ by
\[
q(0)=v_{0}\text{ and }q'(t)=\frac{\d x(t_{\mathrm{pre}})}{\d t}\cdot\frac{t_{\mathrm{next}}-t}{t_{\mathrm{next}}-t_{\mathrm{pre}}}+\frac{\d x(t_{\mathrm{next}})}{\d t}\cdot\frac{t_{\mathrm{pre}}-t}{t_{\mathrm{next}}-t_{\mathrm{pre}}}
\]
with $t_{\mathrm{pre}}=\left\lfloor \frac{t}{h/k}\right\rfloor \cdot\frac{h}{k}$
and $t_{\mathrm{next}}=t_{\mathrm{pre}}+\frac{h}{k}$. Clearly, we
have that $q(0)=x(0)=v_{0}$ and $\frac{\d q}{\d t}(0)=\frac{\d x}{\d t}(0)=v_{1}$.
Also, we have that
\begin{align}
\left\| \frac{\d^{2}q}{\d t^{2}}(t)-\frac{\d^{2}x}{\d t^{2}}(t)\right\| _{2} & =\left\| \int_{t_{\mathrm{pre}}}^{t_{\mathrm{next}}}\frac{\d^{2}x}{\d t^{2}}(s)ds-\frac{\d^{2}x}{\d t^{2}}(t)\right\| _{2}\leq\frac{h}{k}\max_{0\leq t\leq h}\left\| \frac{\d^{3}x}{\d t^{3}}(t)\right\| _{2}\leq\frac{M_{2}h}{k}\max_{0\leq t\leq h}\left\| \frac{\d x}{\d t}(t)\right\| \label{eq:strongly_convex_smooth_q}
\end{align}
where we used that $\frac{\d^{3}x}{\d t^{3}}(t)=-\nabla^{2}f(x(t))\frac{\d x}{\d t}(t)$
at the end. 

Therefore, it suffices to bound the term $\max_{0\leq t\leq h}\left\| \frac{\d x}{\d t}(t)\right\| $.
Using again that $\frac{\d^{3}x}{\d t^{2}}(t)=-\nabla^{2}f(x(t))\frac{\d x}{\d t}(t)$,
we have that
\begin{align*}
\frac{\d x}{\d t}(t) & =\frac{\d x}{\d t}(0)+\frac{\d^{2}x}{\d t^{2}}(0)\cdot t+\int_{0}^{t}(t-s)\cdot\frac{\d^{3}x}{\d t^{3}}(s)ds\\
 & =v_{1}-\nabla f(v_{0})\cdot t-\int_{0}^{t}(t-s)\cdot\nabla^{2}f(x(s))\frac{\d x}{\d t}(s)ds.
\end{align*}
Hence, for $0\leq t\leq h$, we have that
\[
\left\| \frac{\d x}{\d t}(t)\right\| _{2}\leq\|v_{1}\|_{2}+\|\nabla f(v_{0})\|_{2}\cdot h+M_{2}\cdot\int_{0}^{t}(t-s)\left\| \frac{\d x}{\d t}(t)\right\| _{2}ds.
\]
Solving this integral inequality, Lemma \ref{lem:cosh_ODE} shows that 
\[
\max_{0\leq t\leq h}\left\| \frac{\d x}{\d t}(t)\right\| \leq\left(\|v_{1}\|_{2}+\|\nabla f(v_{0})\|_{2}\cdot h\right)\cdot\cosh(\sqrt{M_{2}}\cdot h)\leq2\left(\|v_{1}\|_{2}+\|\nabla f(v_{0})\|_{2}\cdot h\right)
\]
where we used that $h\leq\frac{1}{2\sqrt{M_{2}}}$. Applying this
in (\ref{eq:strongly_convex_smooth_q}) gives this result.
\end{proof}
The precise runtime depends on how accurate we need to solve the ODE,
namely, the parameter $\epsilon$ in Lemma \ref{lem:2nd_poly}. The
term $\|v_{1}\|_{2}$ can be upper bounded by $O(\sqrt{d})$ with high  probability 
since $v_{1}$ is sampled from normal distribution. The term $\|\nabla f(v_{0})\|_{2}$
is much harder to bound. Even for a random $v_0 ~ e^{-f}$, the worst case bound we can give is 
$$\|\nabla f(v_{0})\|_{2} = \|\nabla f(v_{0})-\nabla f(x^*)\|_{2} \leq M_2 \| v_0 - x^* \|_2 \lesssim \sqrt{\kappa M_{2}d}$$
where we used that $\| v_0 - x^* \|_2 \lesssim \sqrt{d/m_2}$ for random $v_0~e^{-f}$ \cite{dm16}. This is not enough for improving existing algorithms, as we would need $\sqrt{\kappa d}$ time per iteration. The crux of
this section is to show that $\|\nabla f(v_{0})\|_{2}=O(\sqrt{M_{2}d})$
for most of the iterations in the HMC walk if the process starts at the minimum of $f$. This is tight for quadratic $f$.
\begin{lemma}[Smoothness implies expected gradient is upper bounded]
\label{lem:amortized_norm_grad} Let $f$ be a function such that 
$- M_2 \cdot I\preceq \nabla^{2}f(x)\preceq M_{2}\cdot I$
for all $x\in\mathbb{R}^{d}$. Let $x^{(k)}$ be the starting point
of the $k^{th}$ step in $\textsc{HMC}(x^{(0)},f,\epsilon,h)$ (Algorithm~\ref{alg:hmc}) with step size $h\leq\frac{1}{8\sqrt{M_{2}}}$.
Then, we have that
\[
\frac{1}{N} \E \left[ \sum_{k=0}^{N-1}\|\nabla f(x^{(k)})\|_{2}^{2} \right] \leq O\left(\frac{f(x^{(0)})-\min_{x}f(x)}{h^{2}N}+M_{2}d + \frac{\overline{\epsilon}^2}{h^4} \right) 
\]
$\ov{\epsilon}$ is the error in solving the HMC defined in Algorithm~\ref{alg:hmc}.
\end{lemma}

\begin{proof}
Consider one step of the HMC dynamic. Note that
\[
\frac{\d}{\d t}f(x(t))=\nabla f(x(t))^{\top}\frac{\d x}{\d t}.
\]
Hence, we have
\begin{align}
\frac{\d^{2}}{\d t^{2}}f(x(t)) 
= & ~ \frac{\d x}{\d t}^{\top}\nabla^{2}f(x(t))\frac{\d x}{\d t}+\nabla f(x(t))^{\top}\frac{\d^{2}x}{\d t^{2}} \nonumber \\
= & ~ \frac{\d x}{\d t}^{\top}\nabla^{2}f(x(t))\frac{\d x}{\d t}-\|\nabla f(x(t))\|^{2} & \text{~by~} \frac{\d^2 x}{\d t^2} = - \nabla f(x(t)) \nonumber \\
\leq & ~  M_{2}\cdot\|\frac{\d x}{\d t}\|_{2}^{2}-\|\nabla f(x(t))\|^{2}. & \text{~by~} \nabla^2 f(x(t)) \preceq M_2 \cdot I \label{eq:HMC_df}
\end{align}
In Lemma \ref{lem:2nd_poly}, we proved that $\|\frac{\d x}{\d t}\|_{2}\leq2\left(\|\frac{\d x}{\d t}(0)\|_{2}+\|\nabla f(x(0))\|_{2}\cdot h\right)$
for all $0\leq t\leq h$. Using this, we have that
\begin{align*}
\| \nabla f( x(t) ) - \nabla f( x(0) ) \| 
\leq & ~ M_2 \| x(t) - x(0) \|_2 \\
\leq & ~ M_2 \int_0^t \| \frac{\d x}{\d t} (t) \|_2 \d t \\
\leq & ~ 2 M_2 h \cdot\left(\|\frac{\d x}{\d t}(0)\|_{2}+\|\nabla f(x(0))\|_{2}\cdot h\right)
\end{align*}
which implies
\begin{equation}
\|\nabla f(x(t))\|=\|\nabla f(x(0))\|\pm2M_{2}h\left(\|\frac{\d x}{\d t}(0)\|+\|\nabla f(x(0))\|\cdot h\right).\label{eq:grad_f_changed}
\end{equation}
Using our choice of $h$, we have that $\|\nabla f(x(t))\|\geq\frac{1}{2}\|\nabla f(x(0))\|-2M_{2}h\cdot\|\frac{\d x}{\d t}(0)\|_{2}.$
Putting these estimates into (\ref{eq:HMC_df}) gives
\begin{align*}
 & ~ \frac{\d^{2}}{\d t^{2}}f(x(t)) \\
\leq & ~ 2M_{2}\left(\|\frac{\d x}{\d t}(0)\|^{2}+\|\nabla f(x(0))\|^{2}h^{2}\right)-\frac{1}{4}\|\nabla f(x(0))\|^{2}+2M_{2}h\cdot\|\nabla f(x(0))\|\cdot\|\frac{\d x}{\d t}(0)\|\\
\leq & ~ 2M_{2}\left(\|\frac{\d x}{\d t}(0)\|^{2}+\|\nabla f(x(0))\|^{2}h^{2}\right)-\frac{1}{4}\|\nabla f(x(0))\|^{2}+2 \cdot ( \frac{1}{8} \|\nabla f(x(0))\| )^2 + 2 \cdot ( 16 M_2 h \|\frac{\d x}{\d t}(0)\| )^2\\
= & ~ ( 2M_{2} + 512 M_2^2 h^2) \|\frac{\d x}{\d t}(0)\|^{2}-( \frac{1}{4} - 2 M_2 h^2 - \frac{1}{32} ) \|\nabla f(x(0))\|^{2} \\
\leq & ~ 10 M_2 \cdot\|\frac{\d x}{\d t}(0)\|^{2}- \frac{1}{8}\|\nabla f(x(0))\|^{2}
\end{align*}
where we used that $h\leq\frac{1}{8\sqrt{M_{2}}}$ at the last two
equations.

Since $\frac{\d x}{\d t}(0)$ is sampled from normal distribution, we
have that
\begin{align*}
\E [ f(x(h)) ] 
\leq & ~ f(x(0))+\int_{0}^{h}(h-t)\frac{\d^{2}}{\d t^{2}}f(x(t))\d t\\
= & ~ f(x(0))+5 M_{2}\cdot h^{2} \cdot \E \left[ \| \frac{\d x}{\d t} x(0) \|^2 \right] -\frac{h^{2}}{16}\|\nabla f(x(0))\|^{2} \\
= & ~ f(x(0))+ 5 M_{2}\cdot h^{2} \cdot d -\frac{h^{2}}{16}\|\nabla f(x(0))\|^{2} 
\end{align*}
where the last step follows from $\E [\| \frac{\d x}{\d t} (0) \|^2] = d$.

For the ODE starting from $x^{(k)}$, we have that
\[
\E [ f(x(h)) ] \leq f(x^{(k)})-\frac{h^{2}}{16}\|\nabla f(x^{(k)})\|^{2}+O(M_{2}\cdot d\cdot h^{2}).
\]
To compare $f(x(h))$ and $f(x^{(k+1)})$, we note that the distance
between $x^{(k+1)}$ and $x(h)$ is less than $\overline{\epsilon}$ in $\ell_2$ norm. 
Using (\ref{eq:grad_f_changed}) and the fact $\nabla^{2}f\preceq M_{2} I$,
the function value changed by at most
\[
\overline{\epsilon} \cdot\|\nabla f(x^{(k+1)})\|+M_{2}\cdot \overline{\epsilon}^{2}= 2\overline{\epsilon} \cdot\|\nabla f(x^{(k)})\|+O(\overline{\epsilon} M_{2}\sqrt{d}h+\overline{\epsilon}^{2}M_{2})
\]
with high probability. Hence, we have
\begin{align*}
\E [ f(x^{(k+1)}) ] 
\leq & ~ f(x^{(k)})+ 2\overline{\epsilon} \|\nabla f(x^{(k)})\|-\frac{h^{2}}{16}\|\nabla f(x^{(k)})\|^{2}+O(\overline{\epsilon} M_{2}\sqrt{d}h+\overline{\epsilon}^{2}M_{2}+M_{2}dh^{2}) \\
\leq & ~ f(x^{(k)})+ 2\overline{\epsilon} \|\nabla f(x^{(k)})\|-\frac{h^{2}}{16}\|\nabla f(x^{(k)})\|^{2}+O(\overline{\epsilon}^{2}M_{2}+M_{2}dh^{2}) \\
\leq & ~ f(x^{(k)})+ 2\overline{\epsilon}\|\nabla f(x^{(k)})\|-\frac{h^{2}}{16}\|\nabla f(x^{(k)})\|^{2}+O( \frac{ \overline{\epsilon}^2} {h^2} + M_{2}dh^{2}) \\
\leq & ~ f(x^{(k)})-\frac{h^{2}}{32}\|\nabla f(x^{(k)})\|^{2}+O(\frac{\overline{\epsilon}^{2}}{h^{2}} + M_{2}dh^{2}) 
\end{align*}
where the step follows from $2ab \leq a^2 + b^2$, the third step follows from our choice of $h$, the second last step follows by $2 \overline{\epsilon} \| \nabla f(x^{(k)}) \| \leq  \frac{h^2}{64} \| \nabla f(x^{(k)}) \|^2 + 64 \frac{\overline{\epsilon}^2}{h^2}$.

Summing $k$ from $0$ to $N-1$, we have
\begin{align*}
\sum_{k=0}^{N-1} \E \left[ f ( x^{(k+1)} ) - f( x^{(k)} ) + \frac{h^2}{32} \| \nabla f ( x^{(k)} ) \|^2 \right] \leq N \cdot O \left( \frac{\ov{\epsilon}^2}{ h^2} +  M_2 d h^2 \right)  
\end{align*}
Using $f( x^{(N)} ) \geq \min_x f(x)$ and reorganizing the terms gives the desired result.

\end{proof}

\begin{remark}
Both Lemma~\ref{lem:2nd_poly} and Lemma~\ref{lem:amortized_norm_grad} do not need convexity.
\end{remark}

\subsection{Sampling}

Now, we can apply our algorithm for second-order ODEs to the HMC dynamic. To control the gradient of the initial point in a simple way, we start at the algorithm at a local minimum of $f$.

\restate{thm:strongly_convex}

\begin{proof}
The number of iterations follows from Theorem \ref{thm:contraction_HMC} with 
\begin{align}\label{eq:our_choice_of_h_and_epsilon}
    h=\frac{m_{2}^{1/4}}{ 16000 M_{2}^{3/4}}.
\end{align}

To approximate the HMC dynamic, we apply the ODE algorithm (Theorem
\ref{thm:kth_order_ode_piecewise}). Now, we estimate the parameters in Theorem \ref{thm:kth_order_ode_piecewise}. Note that $L_{1}=0$,
$L_{2}=M_{2}$, $L=\sqrt{M_{2}}$, $T=h$, $LT\leq1/16000$, $\epsilon_{\mathrm{ODE}} = \frac{\epsilon \cdot \sqrt{m_2} h^2}{16}$. Hence, if the solution
is approximated by a $k$-piece degree 2 polynomial, we can find it
in $O(d k)\cdot\log(\frac{C}{\epsilon_{\mathrm{ODE}}})$ time and $O(k)\cdot\log(\frac{C}{\epsilon_{\mathrm{ODE}}})$
evaluations where 
\[
C=O(h^{2})\|\nabla f(v_{0})\|+h\|v_{1}\|
\]
with $v_{0}$ and $v_{1}$ are the initial point and initial velocity
of the dynamic. 

Finally, Lemma \ref{lem:2nd_poly} shows that the ODE can be approximated
using a 
\[
k\defeq\frac{2 \cdot M_{2}h^{3}\left(\|v_{1}\|_{2}+\|\nabla f(v_{0})\|_{2}\cdot h\right)}{\epsilon/(\sqrt{m_{2}}\cdot\kappa^{3/2})}
\]
piece degree 2 polynomials where $v_{0}$ is the initial point and
$v_{1}$ is the initial velocity. Since $v_{1}$ is sampled from normal
distribution, we have that
\[
k \lesssim \frac{1}{\epsilon} \left(\kappa^{\frac{1}{4}}\sqrt{d}+\frac{\|\nabla f(v_{0})\|_{2}}{\sqrt{M_{2}}}\right).
\]
in expectation. 

Now, we apply Lemma \ref{lem:amortized_norm_grad} with $\overline{\epsilon} = \sqrt{m_2}h^2 \eps/16$ and use the fact
that $f(x^{(0)})-\min_{x}f(x)\leq\frac{\|\nabla f(x^{(0)})\|_{2}^{2}}{m_{2}}$, we have
\begin{align*}
\E \left[ \frac{1}{N}\sum_{k=0}^{N-1}\|\nabla f(x^{(k)})\|_{2}^{2} \right] 
\lesssim & ~  \frac{\|\nabla f(x^{(0)})\|_{2}^{2}}{m_{2}h^{2}N} +  M_{2}d + \frac{ \overline{\epsilon}^2 }{ h^4 } \\
\lesssim & ~ \frac{\|\nabla f(x^{(0)})\|_{2}^{2}}{m_{2}h^{2}N} +  M_{2}d + m_2 \epsilon^{2} \\
\lesssim & ~  \|\nabla f(x^{(0)})\|_{2}^{2}+M_{2}d+m_{2}\epsilon^{2} \\
\lesssim & ~  \|\nabla f(x^{(0)})\|_{2}^{2}+M_{2}d
\end{align*}
where the third step follows from $N \geq \frac{1}{m_2 h^2}$, 
our choice of $h$ and $\epsilon$ (i.e., Eq.~\eqref{eq:our_choice_of_h_and_epsilon}), and we used that $\epsilon\leq\sqrt{d}$ at the end. Hence, the expected
number of evaluations per each HMC iterations (amortized over all HMC iterations) is 
\begin{align*}
O(k)\cdot\log(\frac{C}{\epsilon_{\mathrm{ODE}}}) 
\lesssim & ~ \frac{1}{\epsilon} \left(\kappa^{\frac{1}{4}}\sqrt{d}+\frac{\|\nabla f(x^{(0)}) \| + \sqrt{M_{2}d}}{\sqrt{M_{2}}}\right)\log\left(\frac{h^{2}(\|\nabla f(x^{(0)}) \| + \sqrt{M_{2}d})+h\sqrt{d}}{\epsilon_{\mathrm{ODE}}}\right)\\
\lesssim & ~ \frac{1}{\epsilon} \left(\kappa^{\frac{1}{4}}\sqrt{d}+\frac{\|\nabla f(x^{(0)}) \| }{\sqrt{M_{2}}}\right)\log\left(\frac{1}{\epsilon}\left(\frac{\|\nabla f(x^{(0)})\|}{\sqrt{m_{2}}}+\kappa^{3/4}\sqrt{d}\right)\right)
\end{align*}
where the last step follows from our choice of $\epsilon_{\mathrm{ODE}} = \frac{\epsilon \cdot \sqrt{m_2} h^2}{16}$ and our choice of $h=\frac{m_{2}^{1/4}}{ 16000 M_{2}^{3/4}}$. Since we start at the minimum of $f$, $\|\nabla f(x^{(0)}) \|=0$ and this gives the expected number of evaluations.

Similarly, we have the bound for expected time. This completes the proof.
\end{proof}

\section{Sampling from Incoherent Logistic Loss Functions and More}\label{sec:logistic}

In this section we prove Theorem \ref{thm:sampling_formal}. Since the function $$f(x) = \sum_{i=1}^n \phi_i(a_i^{\top} x) + \frac{m_2}{2} \|x\|^2_2,$$ the HMC dynamic for sampling $e^{-f(x)}$ is given by
\[
\frac{\d^{2}}{\d t^{2}}x(t)=-\nabla f(x(t))=-A^{\top}\phi'(Ax) - m_2 x
\]
where the $i^{th}$ row of $A \in \R^{n \times d}$ is $a_{i}^\top$, $\forall i \in [n]$ and $\phi':\R^n\rightarrow\R^n$ is defined by
\begin{align*}
\phi'(s) = \left( \frac{ \mathrm{d} }{ \mathrm{d} s_1 } \phi_1(s_1) , \frac{ \mathrm{d} }{ \mathrm{d} s_2 } \phi_2(s_2) , \cdots , \frac{ \mathrm{d} }{ \mathrm{d} s_n } \phi_n(s_n) \right).
\end{align*}

To simplify the proof, we let $s(t)=Ax(t)$. Then, $s$ satisfies
the equation
\begin{equation}\label{eq:s_ODE}
\frac{\d^{2}}{\d t^{2}}s(t)=F(s(t))\quad\text{where}\quad F(s)=-AA^{\top}\phi'(s) - m_2 s.
\end{equation}
Ignoring the term $m_2 s$, $F$ consists of two parts the first part $-AA^{\top}$
is linear and the second part is decoupled in each variable. This
structure allows us to study the dynamic $s$ easily. In this section, we discuss
how to approximate the solution of (\ref{eq:s_ODE}) using the collocation
method.

The proof consists of (a) bounding $\|s(t)\|_\infty$, (b) bounding Lipschitz constant of $F$ and
(c) showing that $s(t)$ can be approximated by a polynomial.

\subsection{$\ell_{\infty}$ bound of the dynamic}
\begin{lemma}[$\ell_{\infty}$ bound of the dynamic]\label{lem:l_inf_bound_of_the_dynamic}
Let $x^{(j)}$ be the $j^{th}$ iteration of the HMC dynamic defined
in Algorithm~\ref{alg:hmc} with 
\begin{align*}
f(x)=\sum_{i=1}^{n}\phi_{i}(a_{i}^{\top}x)+\frac{m_{2}}{2}\|x\|_{2}^{2}.
\end{align*}
Assume that $m_{2}\cdot I\preceq\nabla^{2}f(x)\preceq M_{2}\cdot I$
for all $x$, $|\phi'(s)|\leq M$ for all $s$, and that $\tau=\|AA^{\top}\|_{\infty\rightarrow\infty}$.
Let $s^{(j)}=Ax^{(j)}$. Suppose that the step size $T\leq\frac{1}{2\sqrt{m_{2}}}$,
we have that
\[
\max_{ j \in [ N ] }\|s^{(j)}-s^{(0)}\|_{\infty}= O \left( \sqrt{\frac{\tau}{m_{2}}} + \frac{\tau M}{m_{2}} \right) \cdot \log( dN / \eta)
\]
with probability at least $1-\eta$.
\end{lemma}

\begin{proof}
We ignore the index of iteration $(j)$ and focus on how $s$ changes
within each ODE first.

For any $0\leq t\leq T$ and for any $i$, we have that
\begin{equation}\label{eq:s_bound_eq}
s_{i}(t)=s_{i}(0)+s_{i}'(0)t-\int_{0}^{t}(t-\ell)(AA^{\top}\phi'(s(\ell)))_{i} \d \ell-m_{2}\int_{0}^{t}(t-\ell)s_{i}(\ell) \d \ell.
\end{equation}
Using $\tau=\|AA^{\top}\|_{\infty \rightarrow \infty}$ and $|\phi'(s)|\leq M$,
we have
\begin{align*}
|s_{i}(t)-s_{i}(0)| & \leq T\cdot|s_{i}'(0)|+\int_{0}^{T}(T-\ell)\cdot\tau\cdot M \d \ell+m_{2}\int_{0}^{t}(t-\ell)|s_{i}(\ell)-s_{i}(0)| \d \ell+ \frac{1}{2} m_{2} T^{2} |s_{i}(0)|\\
 & =\frac{1}{2} m_{2}T^{2} \cdot |s_{i}(0)|+T\cdot|s_{i}'(0)|+\frac{1}{2} T^{2} \cdot\tau M+m_{2}\int_{0}^{t}(t-\ell)|s_{i}(\ell)-s_{i}(0)| \d \ell.
\end{align*}
Solving this integral inequality (Lemma~\ref{lem:cosh_ODE}), we get
\begin{align}\label{eq:s_i_t_minus_s_i_0}
|s_{i}(t)-s_{i}(0)| 
\leq & ~ \left(\frac{1}{2} m_{2}T^{2} \cdot|s_{i}(0)|+T\cdot|s_{i}'(0)|+\frac{1}{2} T^{2} \cdot\tau M\right)\cdot\cosh(\sqrt{m_{2}}t) \notag \\
\leq & ~ \frac{1}{4} |s_{i}(0)| +2T\cdot|s_{i}'(0)|+T^{2}\cdot\tau M
\end{align}
where we used that $t\leq T\leq\frac{1}{2\sqrt{m_{2}}}$.

Using this estimate back to (\ref{eq:s_bound_eq}), we have 
\begin{align*}
 & ~ \left|s_{i}(T)-s_{i}(0)-s_{i}'(0)T + \frac{1}{2} m_{2}T^{2} s_{i}(0)\right| \\
\leq & ~ \int_{0}^{T}(T-\ell)|AA^{\top}\phi'(s(\ell))_i| \d\ell+m_{2}\int_{0}^{T}(T-\ell)|s_{i}(\ell)-s_{i}(0)| \d \ell\\
\leq & ~ \int_0^T (T - \ell) \tau M \d \ell + m_2 \int_0^T (T - \ell) \left(\frac{1}{4} |s_{i}(0)| + 2T\cdot|s_{i}'(0)|+T^{2}\cdot\tau M\right) \d \ell \\
= & ~ \frac{1}{2} T^{2} \cdot\tau M + \frac{1}{2} m_{2}T^{2} \cdot\left(\frac{1}{4} |s_{i}(0)| + 2T\cdot|s_{i}'(0)|+T^{2}\cdot\tau M\right)\\
\leq & ~ \frac{1}{2} T^{2} \cdot\tau M + \frac{1}{2} m_{2}T^{2} \cdot\left(\frac{1}{4} |s_{i}(0)| + \frac{1}{\sqrt{m_2}} \cdot|s_{i}'(0)| +  \frac{1}{4 m_2} \cdot\tau M\right) \\
= & ~ \frac{5}{8} T^{2}\cdot\tau M+\frac{1}{8} m_{2}T^{2} |s_{i}(0)| + \frac{1}{2} \sqrt{m_{2}}T^{2}\cdot|s_{i}'(0)| 
\end{align*}
where the second step follows from $|A A^\top \phi'(s(\ell))_i|\leq \tau M$ and \eqref{eq:s_i_t_minus_s_i_0}, the third step follows from $\int_0^T (T - \ell) \d \ell = \frac{1}{2} T^2$, the fourth step follows from $T\leq\frac{1}{2\sqrt{m_{2}}}$. 

Note that we bounded how much each iteration of the HMC dynamic can change the solution.
Writing it differently, we have
\begin{align*}
s_{i}^{(j+1)}=(1- \frac{1}{2} m_{2}T^{2} )s_{i}^{(j)}+{s_{i}^{(j)}}'(0)\cdot T+\beta
\end{align*}
where
\begin{align*}
|\beta| \leq \frac{1}{4} m_{2}T^{2} |s_{i}^{(j)} (0)| + \sqrt{m_{2}}T^{2}\cdot|{s_{i}^{(j)}}'(0)| + \tau M \cdot T^2.
\end{align*}

Note that ${s_{i}^{(j)}}'(0)\sim {\cal N}(0,(A^{\top}A)_{i})$ and that
\begin{align*}
\lambda_{\max}(A^{\top}A)=\lambda_{\max}(AA^{\top})\leq\|AA^{\top}\|_{\infty\rightarrow\infty}=\tau.
\end{align*}
Therefore, ${s_{i}^{(j)}}'(0)T\sim \alpha_i \cdot {\cal N}(0,1)$ with $0 \leq \alpha_i \leq \sqrt{\tau}T$. 
Now, we simplify the dynamic to 
\begin{align*}
s_{i}^{(j+1)} = (1-\delta) s_i^{(j)} + \alpha_i N^{(j)} + \beta
\end{align*}
where $\delta := \frac{1}{2} m_{2}T^{2}$, $\alpha_{i} \leq a := \sqrt{\tau}T$, and
\begin{align*}
|\beta | \leq & ~ \frac{\delta}{2} |s_i^{(j)}(0)| + b |N^{(j)}| + c
\end{align*}
with $b:=\sqrt{m_{2}\tau}T^{2}$ and $c:= \tau M \cdot T^2$.

Applying Lemma \ref{lem:s_bound_martingale}, we have that 
\begin{align*}
\Pr \left[ \max_{ j \in [ N ] }|s_{i}^{(j)}-s_{i}^{(0)}|
\geq C \cdot \left( \frac{a}{ \sqrt{\delta} } + \frac{c+ b}{\delta}  \right) \log( N / \epsilon ) \right] \leq\epsilon
\end{align*}
for some constant $C$. Taking union bound
over $i \in [d]$, the bound follows from the calculation
\begin{align*}
\frac{a}{ \sqrt{\delta} } + \frac{c+ b}{\delta} 
= & ~ \frac{ \sqrt{2 \tau}T}{\sqrt{m_{2}}T} + \frac{\tau M \cdot T^2 +\sqrt{m_{2}\tau} \cdot T^{2}}{2 m_{2}T^{2}} \gtrsim  \sqrt{\frac{\tau}{m_{2}}} + \frac{\tau M}{m_{2} }.
\end{align*}
\end{proof}
\begin{lemma}[Bounding the Martingale]\label{lem:s_bound_martingale}
Let $X^{(i)}$ be a sequence of random
variable such that 
\[
X^{(i+1)} = (1-\delta) X^{(i)}+ \alpha_i N^{(i)}+ \beta
\]
where $N^{(i)}\sim {\cal N}(0,1)$ are independent, $\alpha_{i}\leq a$, $\beta\leq\frac{\delta}{2}|X^{(i)}|+b|N^{(i)}|+c$
with positive $a,b,c$ and $0<\delta\leq1$. For some constant universal
$C > 0$, we have that 
\begin{align*}
\Pr \left[ \max_{ i \in [ k ] }|X^{(i)}-X^{(0)}|\geq C\cdot \left( \frac{a}{\sqrt{\delta}}+\frac{c+b}{\delta} \right) \log( k / \epsilon ) \right] \leq \epsilon
\end{align*}
for any $0<\epsilon<1$.
\end{lemma}

\begin{proof}
We will first show that $X^{(i)}$ cannot grow to large. The proof of the other direction is similar. Consider the potential
$\Phi^{(i)}=\E \left[ e^{\lambda X^{(i)}} \right]$. Note that
\begin{align*}
\Phi^{(i+1)}
\leq & ~ \E \left[ e^{\lambda((1-\delta)X^{(i)}+\alpha_i N^{(i)}+b|N^{(i)}|+c+\frac{\delta}{2}|X^{(i)}|)} \right] \\
\leq & ~ e^{\lambda c} \E \left[ e^{\lambda(1-\delta)X^{(i)}+\frac{\delta}{2}|X^{(i)}|} \right] \E \left[ e^{\lambda(\alpha_i  N^{(i)}+b|N^{(i)}|)} \right],
\end{align*}
where we used that $N^{(i)}$ are independent. Picking $\lambda\leq\frac{1}{a+b}$
and using $N^{(i)}\sim {\cal N}(0,1)$ and $|\alpha_i| \leq a$, we have that
\begin{align*}
\E \left[ e^{\lambda(\alpha_i  N^{(i)}+b|N^{(i)}|)} \right] \leq e^{O(\lambda b+\lambda^{2}a^{2})}.
\end{align*}
Therefore, for any $\eta>0$, we have
\begin{align*}
\Phi^{(i+1)} & \leq e^{O(\lambda c+\lambda b+\lambda^{2}a^{2})}\cdot\left(\E  \left[ e^{\lambda(1-\delta)X^{(i)}+\frac{\delta}{2}|X^{(i)}|}1_{X^{(i)}\leq \eta} \right] + \E \left[ e^{\lambda(1-\frac{\delta}{2})X^{(i)}}1_{X^{(i)}>\eta} \right] \right)\\
 & \leq e^{O(\lambda c+\lambda b+\lambda^{2}a^{2})}\cdot\left(e^{\lambda\eta}+\E \left[ e^{\lambda X^{(i)}-\frac{\lambda}{2}\delta\eta} \right] \right)\\
 & \leq e^{O(\lambda c+\lambda b+\lambda^{2}a^{2}+\lambda\eta)}+e^{O(\lambda c+\lambda b+\lambda^{2}a^{2})-\frac{\lambda}{2}\delta\eta}\Phi^{(i)}.
\end{align*}
Choose $\eta=\Theta(\frac{c+b+\lambda a^{2}}{\delta})$ such that
$O(\lambda c+\lambda b+\lambda^{2}a^{2})-\frac{\lambda}{2}\delta\eta\leq0$.
Hence, we have
\[
\Phi^{(i+1)}\leq e^{O(\lambda c+\lambda b+\lambda^{2}a^{2}+\lambda\eta)}+\Phi^{(i)}\leq e^{O(\frac{\lambda c+\lambda b+\lambda^{2}a^{2}}{\delta})}+\Phi^{(i)}.
\]

Picking an appropriate $\lambda$, we have the result.
\end{proof}

\subsection{Lipschitz constant of $F$}
Now, we bound the Lipschitz constant of function $F$.
\begin{lemma}[Lipschitz bound of function $F$]\label{lem:lipschitz_bound_of_function_F}
Let  $L_{\phi'}$ be the Lipschitz constant of $\phi'$, i.e.,
\begin{align*}
\| \phi'(s_1) - \phi'(s_2) \|_{\infty} \leq L_{\phi'} \| s_1 - s_2 \|_{\infty}, \forall s_1, s_2.
\end{align*}
The function $F$ defined in (\ref{eq:s_ODE}) has Lipschitz constant $(L_{\phi'} \tau + m_2)$ in $\ell_\infty$ norm where $ \tau = \| A A^\top \|_{\infty \rightarrow \infty}  $.
\end{lemma}

\begin{proof}
Note that
\begin{align*}
\left\| F( s_1 ) - F( s_2 ) \right\|_{\infty} 
= & ~ \|  A A^\top ( \phi'(s_1) - \phi'(s_2) ) \|_{\infty} + m_2 \| s_1 -s_2 \|_{\infty} \\ 
\leq & ~ \tau \cdot \| \phi'(s_1) - \phi'(s_2) \|_{\infty}  + m_2 \| s_1 -s_2 \|_{\infty} \\
\leq & ~ (\tau \cdot L_{\phi'}+m_2) \cdot \| s_1 - s_2 \|_{\infty} ,
\end{align*}
where the second step follows by $ \| A A^\top \|_{\infty \rightarrow \infty} = \tau $, and second step follows by $\phi'$ is $L_{\phi'}$-Lipschitz function.

\end{proof}

\begin{remark}\label{rem:lipschitz_bound_of_function_F}
If we think of the role of $F$ in our second order ODE \eqref{eq:s_ODE}, then $F$ is in fact independent of $\frac{\d s}{\d t}$. Therefore $L_1 = 0 $ and $L_2 = L_{\phi'} \tau + m_2$.
\end{remark}

\subsection{Existence of low-degree solutions}
Next, we establish bounds on the radius up to which the solution to the ODE (\ref{eq:s_ODE}) have a low-degree polynomial approximation.

\begin{lemma}[Low-degree polynomial approximation]\label{lem:existence_ode}
Assume that $\| A A^\top \|_{\infty \rightarrow \infty} = \tau$ and that for all $i$, $\phi'_i$ has Cauchy estimate $M$ with radius $r$, i.e.,
\begin{align*}
\forall l\geq0, \forall a \in \R, | (\phi_i)^{(l+1)} (a) | \leq M \cdot l ! \cdot r^{-l}.
\end{align*}
Let $s^*(t) \in \R^{n}$ denote the solution of the ODE (\ref{eq:s_ODE}).

For any 
\begin{align*}
0 \leq T \leq \frac{r}{4}\left( ( M\tau r+m_{2}r (r+\|s(0)\|_{\infty}))^{1/2}+\|s'(0)\|_{\infty}\right)^{-1}
\end{align*}
 and any $0<\epsilon<1$,
there is a degree $D = 2 \left\lceil 4+\log_{2}(1/\varepsilon)\right\rceil $ polynomial $q : \R \rightarrow \R^n$ such that 
\begin{align*}
q(0) = s^*(0), q'(0) = {s^*}'(0),\text{~and~} \left\| \frac{ \mathrm{d}^2 }{ \mathrm{d} t^2 } q(t) - \frac{ \mathrm{d}^2 }{ \mathrm{d} t^2 } s^*(t) \right\|_{\infty} \leq \frac{ \epsilon \cdot r }{ T^2 }, \forall t \in [0,T]
\end{align*}
\end{lemma}


%
%
%
%
First, we verify the condition in Lemma~\ref{lem:bound_D_k_F_u}.

\define{lemmaF}{Lemma}{{\rm(Bounding derivatives of $F$)}{\bf.}
Under the same assumptions in Lemma \ref{lem:existence_ode}, we have
\begin{align*}
\| D^k F( s ) [ \Delta_1, \Delta_2, \cdots, \Delta_k ]  \|_{\infty} \leq g^{(k)} (0) \cdot \prod_{j=1}^k \| \Delta_j \|_{\infty} , \forall k \geq 0, \forall \Delta_1, \cdots, \Delta_k.
\end{align*}
where $$g(x)= \frac{\tau \cdot M + m_2 \cdot (r + \| s \|_\infty) }{1 - r^{-1} x}.$$
}
\state{lemmaF}

\begin{proof}
Recall that $\phi'(x) = ( \phi_1'(x_1) , \phi_2'(x_2) , \cdots, \phi'_n(x_n) )$, $\forall x \in \R^n$. We have
\begin{align*}
  & ~ \| A A^\top D^k \phi'( s ) [ \Delta_1, \Delta_2, \cdots , \Delta_k ] \|_{\infty} \\
 \leq & ~ \tau \cdot \| D^k \phi'( s ) [ \Delta_1, \Delta_2, \cdots, \Delta_k ] \|_{\infty} & \text{~by~} \| A A^\top \|_{\infty \rightarrow \infty} = \tau \\
 \leq & ~ \tau \cdot \max_{i \in [n]} | \phi_i^{ ( k + 1 ) } | \cdot | \Delta_{1,i} | \cdot | \Delta_{2,i} | \cdots | \Delta_{k,i} | \\
 \leq & ~ \tau \cdot \max_{i \in [n]} | \phi_i^{ ( k + 1 ) } | \prod_{j=1}^k \| \Delta_j \|_{\infty} \\
 \leq & ~ \tau \cdot M \cdot k! \cdot r^{-k} \cdot \prod_{j = 1}^k \| \Delta_j \|_{\infty}.
\end{align*}

Next, the derivatives of the $m_2 s$ term in $F(s) = - A A^\top \phi'(s) - m_2 s$ can be bounded by the derivatives of $x \rightarrow m_2 (\|s\|_\infty + x)$, which then can be bounded by the derivatives of $x  \rightarrow m_2 (r + \|s\|_\infty)/(1-r^{-1}x)$. This explains the second part of the function $g$.
\end{proof}

\begin{proof}[Proof of Lemma \ref{lem:existence_ode}]

Theorem \ref{thm:bound_2ODE} and Lemma \ref{lemmaF} shows that
\begin{equation}
\|s^{*(k)}(0)\|_{\infty}\leq\frac{k!\alpha^{k}}{r^{-1}}=r\cdot k!\cdot\alpha^{k}\label{eq:s_k_bound}
\end{equation}
where 
\begin{align*}
\alpha= \max\left( \left( \frac{4}{3}\cdot(M\tau+m_{2}(r+\|s(0)\|_{\infty}))\cdot r^{-1} \right)^{1/2} , 2\cdot\|s'(0)\|_{\infty}r^{-1}\right).
\end{align*}
Since $s^{*}$ is real analytic at $0$ (Theorem \ref{thm:bound_2ODE}), around $t=0$,
we have
\[
s^{*}(t)=\sum_{k=0}^{\infty}\frac{s^{*(k)}(0)}{k!}t^{k}.
\]
Apply Theorem \ref{thm:bound_2ODE} repeatedly at every $t$ such that $s^{*}(t)$ is
defined, we can show that the above equation holds as long as the
right hand side converges.

Let $q(t)=\sum_{k=0}^{D}s^{*(k)}(0)t^{k}$. Then, we have that
\begin{align*}
\left\| \frac{ \d^{2}}{ \d t^{2}}q(t)-\frac{\d^{2}}{\d t^{2}}s^{*}(t)\right \| _{\infty} & \leq\left\| \sum_{k=D+1}^{\infty}\frac{s^{*(k)}(0)}{(k-2)!}t^{k-2}\right\| _{\infty}\\
 & \leq r\cdot\sum_{k=D+1}^{\infty}\frac{k!}{(k-2)!}\alpha^{k}t^{k-2}\\
 & \leq \frac{r}{T^2}\cdot\sum_{k=D+1}^{\infty}\frac{(k-1)k}{2^{k}}\\
 & = \frac{r}{T^2} \cdot2^{-D} (D^2 + 3 D + 4)\\
 & \leq 16 \frac{2^{-D/2}r}{T^{2}}
\end{align*}
where we used (\ref{eq:s_k_bound}) at the second inequality, $\alpha T\leq\frac{1}{2}$
at the third inequality, and $D\geq1$ at the fourth
inequality.
\end{proof}

\subsection{Main result}

\begin{theorem}[Formal version of Theorem~\ref{thm:sampling_informal}]\label{thm:sampling_formal}

Let $A=[a_{1};a_{2};\cdots;a_{n}]\in\R^{n\times d}$, $\phi_{i}:\R\rightarrow\R$
be a function, and 
\[
f(x)=\sum_{i=1}^{n}\phi_{i}(a_{i}^{\top}x)+\frac{m_{2}}{2}\|x\|^{2}.
\]
Let $\tau=\|AA^{\top}\|_{\infty\rightarrow\infty}$ and suppose that

1. $f$ has $M_{2}$ Lipschitz gradient,
i.e., $\nabla^{2}f(x)\preceq M_{2}\cdot I$ for
all $x$, 

2. $\phi_{i}'$ has Cauchy estimate $M$ with radius $r$, i.e., $\forall i\in[n],\ell\geq1,s\in\R$,
$|\phi_{i}^{(l+1)}(s)|\leq M\cdot l!\cdot r^{-l}$.

Starting at $x^{(0)}$, we can output a random
point $X$ such that
\[
\E_{Y\varpropto e^{-f}} \left[ \|X-Y\|_{2}^{2} \right] \leq \frac{\epsilon}{\sqrt{m_{2}}}
\]
using 
$N \lesssim k \log \left(\frac{k}{\epsilon} \left(\frac{\|\nabla f(x^{(0)}) \|^2}{m_2} + d \right)\right)$
iterations with 
$$k\lesssim \kappa^{1.5} + \frac{M \tau}{m_2 r} \log(d N) + \frac{\tau}{m_2 r^2}  \log^2(d N) \quad \text{with} \quad \kappa=\frac{M_{2}}{m_{2}}.$$
Each iteration takes
$O(d\log^{3}(\frac{1}{\delta}))$ time and $O(\log^{2}(\frac{1}{\delta}))$ 
evaluations to the function $\phi'$ and the matrix vector multiplications for $A$ and $A^\top$, with 
$$\delta=\Omega(\frac{1}{r})\sqrt{\frac{\lambda_{\min}(A^{\top}A)}{n\cdot m_{2}}}\cdot\frac{\epsilon}{k}.$$

\end{theorem}

\begin{proof}
The proof consists of bounding the cost of each HMC step in Algorithm~\ref{alg:hmc} and bounding the number of steps. 

\textbf{Cost per iteration:}

As we discussed in this section, we consider the ODE \eqref{eq:s_ODE} instead. We will use Theorem~\ref{thm:kth_order_ode_piecewise} to solve the ODE \eqref{eq:s_ODE}. Hence, we need to bound the parameters of Theorem~\ref{thm:kth_order_ode_piecewise}, which are summarized in Table~\ref{tab:summary_of_parameters_theorem_ode}.

\begin{table}[t]
\begin{center}
    \begin{tabular}{ | l | l | l |}
    \hline
    Parameters & Value & Source  \\ \hline
    $k$ & $2$ & Eq.~\eqref{eq:s_ODE}  \\ \hline
    $L_1$ & $0$ & Lemma~\ref{lem:lipschitz_bound_of_function_F} and Remark~\ref{rem:lipschitz_bound_of_function_F} \\ \hline
    $L_2$ & $ \frac{M \tau}{r} + m_2$ & Lemma~\ref{lem:lipschitz_bound_of_function_F} and Remark~\ref{rem:lipschitz_bound_of_function_F} \\ \hline
    $D$ & $  O( \log (1/\delta) )$ & Lemma~\ref{lem:existence_ode} \\ \hline
    $C$ & $ O(r)$& Eq.~\eqref{eq:C_upper_bounded_by_O_r} \\ \hline
    $\epsilon_{\mathrm{ODE}}$ & $O(\delta r)$ & Eq.~\eqref{eq:eps_delta_r} \\ \hline
    \end{tabular}
\end{center}
\caption{Summary of parameters of Theorem~\ref{thm:kth_order_ode_piecewise}}\label{tab:summary_of_parameters_theorem_ode}
\end{table}

\textbf{Parameters $D$ and $\epsilon_{\mathrm{ODE}}$:}
Let $s$ denote its solution. Lemma~\ref{lem:existence_ode} shows that if
\begin{equation}\label{eq:h_condition_1}
h\leq\frac{r}{4}\left( \left( M\tau\cdot r+m_{2}r(r+\|s(0)\|_{\infty}) \right)^{1/2} + \|s'(0)\|_{\infty}\right)^{-1},
\end{equation}
for any $\delta>0$, there is a degree $O( \log ( 1 / \delta ) )$
polynomial $q$ such that
\begin{align}
q(0)=s(0),\quad q'(0)=s'(0),\quad\text{and}\quad\left\| \frac{\d^{2}}{\d t^{2}}q(t)-\frac{ \d^{2}}{\d t^{2}}s(t)\right\| _{\infty}\leq\frac{\delta\cdot r}{T^{2}}\text{ for }t\in[0,h].\label{eq:eps_delta_r}
\end{align}

\textbf{Parameter $C$:}
To apply Theorem~\ref{thm:kth_order_ode_piecewise}, we first show parameter $C \leq O(r)$ as follows:
\begin{align}\label{eq:C_upper_bounded_by_O_r}
C \lesssim & ~ h\cdot\left(h \left\| F(s(0)) \right\| _{\infty}+\|s'(0)\|_{\infty}\right) \notag \\
 \leq & ~ h \cdot \left(  h \left( \left\| AA^{\top}\phi'(s(0))\right\|_{\infty} + \| m_2 s(0) \|_{\infty} \right) +\|s'(0)\|_{\infty} \right) \notag \\
 \leq & ~ h^2 \tau \left\| \phi'(s(0))\right\| _{\infty} + h^2 m_2 \| s(0) \|_{\infty} + h\|s'(0)\|_{\infty} \notag \\
 \leq & ~ h^2 \tau M + h^2 m_2 \| s(0) \|_{\infty} + h \|s'(0)\|_{\infty} \notag \\
 \leq & ~ \frac{r}{16} + \frac{r}{16} + \frac{r}{4} \leq r.
\end{align}
where the second step follows from \eqref{eq:s_ODE} and triangle inequality, the third step follows from \\$\| A A^\top \|_{\infty \rightarrow \infty} = \tau$, the fourth step follows from $\phi'$ has Cauchy estimate $M$ with radius $r$, and the last step follows from \eqref{eq:h_condition_1}.

Now, Theorem~\ref{thm:kth_order_ode_piecewise} shows that if $h$ satisfies (\ref{eq:h_condition_1})
and that
\begin{align}\label{eq:h_condition_2}
h \leq  \frac{1}{16000} \frac{1}{L}
=  \frac{1}{16000} \frac{1}{ L_1 + \sqrt{L_2} } 
=  \frac{1}{16000}\sqrt{\frac{r}{M\tau + m_2 r}},
\end{align}
then, we can find $p$ such that
\begin{equation}
\left\| s(h)-p\right\| _{\infty}\leq O(\delta\cdot r)\label{eq:HMC_1_step_err}
\end{equation}
using $O(\log(\frac{1}{\delta})\log(\frac{C}{\delta\cdot r}))=O(\log^{2}(\frac{1}{\delta}))$
evaluations of $\phi_{i}'$ and $O(d\log^{2}(\frac{1}{\delta})\log(\frac{C}{\delta\cdot r}))=O(d\log^{3}(\frac{1}{\delta}))$
time. 

To understand the condition in Eq.~\eqref{eq:h_condition_1}, we note that $s'(0)=Av$ where
$v\sim {\cal N}(0,I)$. Hence, ${s^{*}}'(0)\sim {\cal N} (0,A^{\top}A)$. Note that
$\lambda_{\max}(A^{\top}A)=\lambda_{\max}(AA^{\top})\leq\|AA^{\top}\|_{\infty\rightarrow\infty}=\tau$.
Hence, we have that 
\[
\|s'(0)\|_{\infty}=O(1)\cdot\sqrt{\tau\cdot\log ( d N / \eta )}
\]
with probability at least $1-\eta$ probability for all $N$ iterations.
In Lemma~\ref{lem:l_inf_bound_of_the_dynamic}, we proved that
\begin{align*}
\|s(0)\|_{\infty}=O\left(\sqrt{\frac{\tau}{m_{2}}}+\frac{\tau M}{m_{2}}\right)\cdot \log( dN / \eta )
\end{align*}
with probability at least $1-\eta$ for all $N$ iterations.

Putting the bound on $\|s'(0)\|$ and $\|s(0)\|$ into the right hand side of Eq.~\eqref{eq:h_condition_1} gives
\begin{align}\label{eq:decided_h}
&\frac{1}{r}\left(\sqrt{M\tau\cdot r+m_{2}r(r+\|s(0)\|_{\infty})}+\|s'(0)\|_{\infty}\right) \\
 \lesssim & ~ \sqrt{\frac{M\tau}{r}}+\sqrt{m_{2}}+ \sqrt{\frac{m_{2}}{r} \left(\sqrt{\frac{\tau}{m_{2}}}+\frac{\tau M}{m_{2}}\right) \log( d N / \eta ) }  +\frac{\sqrt{\tau}}{r} \sqrt{\log (d N /\eta)} \notag \\
 \leq & ~ \frac{M_{2}^{3/4}}{m_{2}^{1/4}} +\left( \frac{ ( \tau m_2 )^{1/4}}{r^{1/2}} + \sqrt{ \frac{M \tau}{r} } + \frac{\sqrt{\tau}}{r} \right)\sqrt{\log( d N / \eta )} \notag \\
 \leq & ~ \frac{M_{2}^{3/4}}{m_{2}^{1/4}} +\left( \frac{\sqrt{\tau}}{r} \sqrt{ \log ( d N / \eta) } + \frac{\sqrt{m_2}}{ \sqrt{\log (d N / \eta)} } + \sqrt{ \frac{M \tau}{r} } + \frac{\sqrt{\tau}}{r} \right)\sqrt{\log( d N / \eta )} \notag \\
 \lesssim & ~ 
 \frac{M_{2}^{3/4}}{m_{2}^{1/4}} + \frac{\sqrt{\tau}}{r}  \log ( d N / \eta) + 
  \sqrt{\frac{M\tau}{r}}\sqrt{\log( d N / \eta )}  \notag
\end{align}
where the second step follows by Eq.~\eqref{eq:h_condition_1}, the third step follows by $\sqrt{m_2} \leq \frac{ M_2^{3/4} }{ m_2^{1/4} }$ and $\sqrt{ \log (d N / \eta)} \geq 1$ , the fourth step follows by $ab \leq a^2 + b^2$, the fifth step follows by $\sqrt{ \log (d N / \eta)} \geq 1$ and $\sqrt{m_2} \leq \frac{ M_2^{3/4} }{ m_2^{1/4} }$.

Therefore, 
\begin{align}\label{eq:we_pick_h}
h = \Theta \left( \sqrt{\frac{M\tau}{r}}\sqrt{\log( d N / \eta )} + \frac{\sqrt{\tau}}{r} \log ( d N / \eta) + \frac{M_{2}^{3/4}}{m_{2}^{1/4}} \right)^{-1}
\end{align} 
satisfies the condition in Eq.~\eqref{eq:h_condition_1} and Eq.~\eqref{eq:h_condition_2}. It also satisfies the condition in Theorem~\ref{thm:contraction_HMC} ($h \leq \frac{m_{2}^{1/4}}{2 M_{2}^{3/4}}$).

Next, we note that the corresponding HMC dynamic $x^{*}(h)$ is given
by
\[
x^{*}(h)=(A^{\top}A)^{-1}A^{\top}s^{*}(h).
\]
Let $p$ be the approximate of $s^{*}(h)$ we find using Theorem~\ref{thm:contraction_HMC} and
$q=(A^{\top}A)^{-1}A^{\top}p$, then, we have
\begin{align*}
  \| x^{*}(h) - q \| _{2} 
 = & ~ \| (A^\top A)^{-1} A^\top s^*(h) - (A^\top A)^{-1} A p \|_2 \\
\leq & ~ \| (A^\top A)^{-1} A^\top \|_{2 \rightarrow 2} \cdot \| s^*(h) - p \|_2 \\
\leq & ~ \|(A^{\top}A)^{-1}A^{\top}\|_{2 \rightarrow 2}\cdot O(\sqrt{n}\cdot r\cdot\delta) \\
\leq & ~ (\lambda_{\min}(A^{\top}A))^{-\frac{1}{2}}\cdot O(\sqrt{n}\cdot r\cdot\delta),
\end{align*}
where the first step follows by definition of $x^*(h)$ and $q$, the second step follows by definition of $\| \|_{2 \rightarrow 2}$ norm, the third step follows from $\| s^*(h) - p \|_2 \leq O(\sqrt{n} r \delta)$ (implied by \eqref{eq:HMC_1_step_err} and $\| \|_2 \leq \sqrt{n} \| \|_{\infty}$).

Using such $q$ as an approximation of $x^{*}(h)$ in Algorithm~\ref{alg:hmc},
Theorem~\ref{thm:contraction_HMC} shows that the $W_{2}$ error of the sampled point is bounded by
\[
O\left((\lambda_{\min}(A^{\top}A))^{-\frac{1}{2}}\cdot\sqrt{n}\cdot r\cdot\delta\right) \leq \frac{ \epsilon \cdot \theta }{ 2 \sqrt{ m_2 } }.
\]
where $\theta = \frac{m_2 h^2}{8}$ and the last step follows from picking
\begin{align*}
\delta = c\cdot\frac{1}{r}\sqrt{\frac{\lambda_{\min}(A^{\top}A)}{n\cdot m_{2}}}\cdot \epsilon \cdot \theta
\end{align*}
for some small enough $c$. The cost of each iteration follows from Theorem~\ref{thm:kth_order_ode_piecewise} and all parameters we pick.

\textbf{Number of iterations:}

Theorem~\ref{thm:contraction_HMC} shows that the number of iterations is
$$O\left(\frac{1}{\theta}\right)\cdot\left(\log\left(\frac{1}{\theta\cdot\epsilon}\right)+\log\left(\frac{\|\nabla f(x^{(0)})\|^2}{m_{2}}+d\right)\right).$$
Finally, we bound the term $1/\theta$ as follows
\begin{align*}
\frac{1}{\theta} & \lesssim\frac{1}{m_{2}}\left(\frac{M\tau}{r}\log\left(\frac{dN}{\eta}\right)+\frac{\tau}{r^{2}}\log^{2}\left(\frac{dN}{\eta}\right)+\frac{M_{2}^{3/2}}{\sqrt{m_{2}}}\right)\\
 & =\kappa^{1.5}+\frac{M\tau}{m_{2}r}\log\left(\frac{dN}{\eta}\right)+\frac{\tau}{m_{2}r^{2}}\log^{2}\left(\frac{dN}{\eta}\right).
\end{align*}
where we used $\kappa = M_2 / m_2$ and \eqref{eq:we_pick_h}.
\end{proof}




\subsubsection*{Acknowledgement}
This work was supported in part by NSF Awards CCF-1740551, CCF-1749609, DMS-1839116, CCF-1563838, CCF-1717349, and E2CDA-1640081.

\appendix

\newpage
\addcontentsline{toc}{section}{References}
\bibliographystyle{alpha}
\bibliography{ref}
\newpage

\section*{Appendix}
\section{Preliminaries}

\subsection{Notation}
For any function $f$, we define $\wt{O}(f)$ to be $f\cdot \log^{O(1)}(f)$. In addition to $O(\cdot)$ notation, for two functions $f,g$, we use the shorthand $f\lesssim g$ (resp. $\gtrsim$) to indicate that $f\leq C g$ (resp. $\geq$) for an absolute constant $C$. 

\begin{definition}[$p \rightarrow q$ norm]
Given matrix $A \in \R^{n \times d}$, we define $\| \|_{p \rightarrow q}$ norm as follows
\begin{align*}
\| A \|_{p \rightarrow q} = \max_{ x \in \R^d } \frac{ \| A x \|_q }{ \| x \|_p }.
\end{align*}
\end{definition}
$\| A \|_{\infty \rightarrow \infty}$ is a special case where $p = \infty$ and $q = \infty$.

\begin{definition}[Wasserstein distance]
The $k$-th Wasserstein distance between two probability measure $\mu$ and $\nu$ is
\begin{align*}
W_k ( \mu , \nu ) = \left( \inf_{ (X,Y) \in {\cal C} (\mu,\nu) } \E \left[ \| X - Y \|^k \right] \right)^{1/k},
\end{align*}
where ${\cal C}(\mu,\nu)$ is the set of all couplings of $\mu$ and $\nu$.
\end{definition}


\begin{definition}[Cauchy's estimates]
We say function $\phi$ has Cauchy estimate $M$ and radius of convergence $r$, if for all $x\in \R$ and for all integers $l \geq 0$
\begin{align*}
| \phi^{(l)}( x ) | \leq M \cdot l! \cdot r^{-l}.
\end{align*}
\end{definition}

\begin{lemma}\label{lem:tau}
If the number of non-zeros for each row of $AA^{\top}$ is bounded
by $s$, then
\[
\lambda_{\max}(AA^{\top})\leq\|AA^{\top}\|_{\infty\rightarrow\infty}\leq\sqrt{s}\cdot\lambda_{\max}(AA^{\top}).
\]
\end{lemma}

\begin{proof}
Let $v$ be the maximum eigenvalue of $AA^{\top}$. Then, we have
that
\[
AA^{\top}v=\lambda_{\max}(AA^{\top})\cdot v.
\]
Since $\|AA^{\top}\|_{\infty\rightarrow\infty}=\max_{\|v\|_{\infty}=1}\|AA^{\top}v\|_{\infty}$,
we have that $\lambda_{\max}(AA^{\top})\leq\|AA^{\top}\|_{\infty\rightarrow\infty}$.

For the another direction, 
\begin{align*}
\|AA^{\top}\|_{\infty\rightarrow\infty} & =\max_{i}\sum_{j}|AA^{\top}|_{ij}\leq\max_{i}\sqrt{s}\cdot\sqrt{\sum_{j}(AA^{\top})_{ij}^{2}}\\
 & \leq\sqrt{s}\cdot\max_{i}\max_{\|v\|_{2}=1}e_{i}^{\top}AA^{\top}v=\sqrt{s}\cdot\lambda_{\max}(AA^{\top}).
\end{align*}
\end{proof}

\subsection{Simple ODEs}

We prove two helper lemmas (Lemma~\ref{lem:cosh_ODE} and \ref{lem:exp_ODE}) for the later use.
\begin{lemma}\label{lem:cosh_ODE}
Given a continuous function $v(t)$ and positive
scalars $\beta,\gamma$ such that 
\[
0 \leq v(t) \leq \beta + \gamma \int_{0}^{t}(t-s)v(s) \d s.
\]
We have that $v(t)\leq\beta\cosh(\sqrt{\gamma}t)$ for all $t \geq 0$.
\end{lemma}

\begin{proof}
Let $\overline{v}(t)$ be the solution of the integral equation $\overline{v}(t)=\beta+\gamma\int_{0}^{t}(t-s)\overline{v}(s) \d s.$
Note that $\overline{v}$ satisfies the ODE 
\[
\overline{v}''(t)=\gamma\overline{v}(t),\quad\overline{v}'(0)=0,\quad\overline{v}(0)=\beta.
\]
Solving it, we have $\overline{v}(t)=\beta\cosh(\sqrt{\gamma}t)$.
Hence, it suffices to prove that $v(t)\leq\overline{v}(t)$ for all
$t\geq0$.

Fix any $\epsilon>0$. We let $T$ be the supremum such that $(1+\epsilon)\overline{v}(t)\geq v(t)$
for all $0\leq t\leq T$. Suppose $T<+\infty$. Then, we have that
\begin{align*}
v(T) \leq\beta+\gamma\int_{0}^{T}(T-s)v(s) \d s \leq\beta+(1+\epsilon)\gamma\int_{0}^{T}(T-s)\overline{v}(s) \d s <(1+\epsilon)\overline{v}(T).
\end{align*}
By the continuity of $v$ and $\overline{v}$, we show that $T$ is
not the supremum. This is a contradiction. Therefore, $T=+\infty$
for any $\epsilon>0$.
\end{proof}

\begin{lemma}\label{lem:exp_ODE}
Given a continuous function $v(t)$ and positive
scalars $\beta,\gamma$ such that 
\[
0 \leq v(t) \leq \beta + \gamma \int_{0}^{t} v(s) \d s.
\]
We have that $v(t)\leq\beta e^{\gamma t}$ for all $t \geq 0$.
\end{lemma}
\begin{proof}
The proof is identical to Lemma~\ref{lem:cosh_ODE}.

Let $\overline{v}(t)$ be the solution of the integral equation $\overline{v}(t)=\beta+\gamma\int_{0}^{t} \overline{v}(s) \d s.$
Note that $\overline{v}$ satisfies the ODE 
\[
\overline{v}''(t)=\gamma\overline{v}(t),\quad\overline{v}'(0)=0,\quad\overline{v}(0)=\beta.
\]
Solving it, we have $\overline{v}(t)=\beta\exp(\gamma t)$.
Hence, it suffices to prove that $v(t)\leq\overline{v}(t)$ for all
$t\geq0$.

Fix any $\epsilon>0$. We let $T$ be the supremum such that $(1+\epsilon)\overline{v}(t)\geq v(t)$
for all $0\leq t\leq T$. Suppose $T<+\infty$. Then, we have that
\begin{align*}
v(T)  \leq\beta+\gamma\int_{0}^{T} v(s) \d s \leq\beta+(1+\epsilon)\gamma\int_{0}^{T} \overline{v}(s) \d s  <(1+\epsilon)\overline{v}(T).
\end{align*}
By the continuity of $v$ and $\overline{v},$we show that $T$ is
not the supremum. This is a contradiction. Therefore, $T=+\infty$
for any $\epsilon>0$.
\end{proof}

\subsection{Cauchy Estimates and Method of Majorants}\label{sec:low_degree_solutions}

In order to prove the solution of the HMC dynamic can be approximated by a low degree polynomial, we first give a general bound the $k^{th}$ derivative of the second order ODE $u''(t) = F(u(t))$, by reducing it to bounding derivatives of a one-variable ODE. The low degree result then follows from: first, we take the Taylor expansion of the original function; second, truncate it at a certain degree; finally we can claim that the low-degree truncation provides a good approximation to the original function.

\begin{theorem}\label{thm:bound_2ODE}
Given vectors $v_{1},v_{0}\in\R^{d}$ and any norm $\|\cdot \|$ on $\R^d$. Let $U\subset\R^{d}$ be a
neighborhood of $v_{0}$ and that $F:U\rightarrow\R^{d}$ is real
analytic near $v_{0}$. Suppose that 
\[
\|D^{(k)}F(v_{0})[\Delta_{1},\Delta_{2},\cdots,\Delta_{k}]\|\leq k! \cdot a \cdot c^k \prod_{j=1}^{k}\|\Delta_{j}\|\text{ for all }k\geq0.
\]
Then, the
ODE
\begin{equation}\label{eq:2ODE}
\frac{\d^{2}}{\d t^{2}}u(t)=F(u(t)),\frac{\d}{\d t}u(0)=v_{1},u(0)=v_{0}
\end{equation}
has a unique real analytic solution around $t=0$. Furthermore, we have
\[
\|u^{(k)}(0)\|\leq\frac{k!\alpha^{k}}{c}\quad\forall k\geq1
\]
with $\alpha=\max(\sqrt{\frac{4}{3}ac},2\| v_1 \| c)$.
\end{theorem}

Theorem \ref{thm:bound_2ODE} involves two steps. The first step (Lemma \ref{lem:bound_D_k_F_u}) involves bounding the derivatives of the solution of the multivariate ODE (\ref{eq:2ODE}) by its scalar version. The second step (Lemma \ref{lem:solve_single_variable_ODE_better}) involves bounding the scalar ODE directly.

The first step follows directly from the majorization proof of Cauchy--Kowalevski theorem. See \cite{v03} for an introduction of the method of majorants. This theorem usually stated qualitatively without an explicit bound. For completeness, we include a proof for the first step.

\begin{theorem}[Cauchy--Kowalevski theorem \cite{c42,k75,k83,n94}]\label{thm:cauchy_kowalevski}
Given vectors $v_{k-1},$ $v_{k-2},$ $\cdots,v_{0} \in \R^{d}$. Let $U\subset\R^{kd+1}$
be a neighborhood of $z\defeq(v_{k-1},v_{k-2},\cdots,v_{0},0)$ and
that $F:U\rightarrow\R^{d}$ is a real analytic near $z$. Then, the
ODE
\begin{align*}
\frac{d^{k}}{\d t^{k}}x(t) & =F\left(\frac{\d^{k-1}}{\d t^{k-1}}x(t),\cdots,x(t),t \right),\\
\frac{d^{i}}{\d t^{i}}x(0) & =v_{i},\forall i\in\{k-1,\cdots,1,0\}
\end{align*}
has a unique real analytic solution around $t=0$.
\end{theorem}

\define{lem:bound_D_k_F_u}{Lemma}{{\rm(Bounding multivariate ODE by scalar ODE)}{\bf.}
Given vectors $v_{1},v_{0}\in\R^{d}$ and any norm $\|\cdot \|$ on $\R^d$. Let $U\subset\R^{d}$ be a
neighborhood of $v_{0}$ and that $F:U\rightarrow\R^{d}$ is a real
analytic near $v_{0}$. Suppose that 
\[
\|D^{(k)}F(v_{0})[\Delta_{1},\Delta_{2},\cdots,\Delta_{k}]\|\leq f^{(k)}(0)\prod_{j=1}^{k}\|\Delta_{j}\|\text{ for all }k\geq0
\]
for some real analytic function $f$ around $0$. Then the ODE (\ref{eq:2ODE}) has a unique real analytic solution around $t=0$. Furthermore, for
any $b \geq \| v_1 \|$, we have
\[
\|u^{(k)}(0)\|\leq\psi^{(k)}(0) \quad \forall k \geq 1
\]
where $\psi$ is the solution of the ODE
\[
\psi''(t)=f(\psi(t)),\psi'(0)= b,\psi(0)=0.
\]
}

\state{lem:bound_D_k_F_u}

For many functions $F$, including the HMC dynamic or complex analytic $F$, we can pick bound the derivatives of $F$ by the function $f(x)=\frac{a}{1-cx}$ for some $a$ and $c$ in Lemma \ref{lem:bound_D_k_F_u}. Therefore, we only need to give a bound on the scalar ODE for this function $f$.

\define{lem:solve_single_variable_ODE_better}{Lemma}{{\rm(Bounding scalar ODE)}
Let $f(x) = \frac{ a }{ 1 - c x }$ with positive $a$ and $c$. Let $\psi(t)$ denote the solution of the ODE
\begin{align*}
\psi''(t) = f( \psi(t) ), \psi'(0) = b, \psi(0) = 0
\end{align*}
with $b\geq0$. Then,
\[
\psi^{(k)}(0)\leq\frac{k!\alpha^{k}}{c}\quad\forall k\geq1
\]
with $\alpha=\max(\sqrt{\frac{4}{3}ac},2bc)$.
}
\state{lem:solve_single_variable_ODE_better}

Finally, we note that Theorem \ref{thm:bound_2ODE} follows from Lemma \ref{lem:bound_D_k_F_u} and Lemma \ref{lem:solve_single_variable_ODE_better} with $f(x)=\frac{a}{1-cx}$.
\section{Deferred Proof for ODE (Section~\ref{sec:ode})}\label{sec:app_ode}

\subsection{Proof of general $k$-th order ODE}\label{sec:app_kth_order_ode}

The goal of this section is to prove Theorem~\ref{thm:kth_order_ode}.
\begin{proof}
We define $\ov{x}(t) \in \R^{k d}$,
\begin{align*}
\ov{x}(t) = ( \ov{x}_1(t) , \ov{x}_2(t), \cdots, \ov{x}_k(t) ) = \left( c_1 \frac{ \d^{k-1} }{ \d t^{k-1} } x(t), c_2 \frac{ \d^{k-2} }{ \d t^{k-2} } x(t) , \cdots, c_k x(t) \right) ,
\end{align*}
where $\ov{x}_1(t) \in \R^d$, $\ov{x}_2(t) \in \R^d$, $\cdots$, $\ov{x}_k(t) \in \R^d$.
We define the norm on $\R^{k d}$ by $\| \overline{x} \| = \sum_{i=1}^k \| \overline{x}_i \|$.

Then we have 
\begin{align*}
\frac{ \d }{ \d t } \ov{x}(t) = \left( c_1 \frac{ \d^k }{ \d t^k } x(t) , c_2 \frac{ \d^{k-1} }{ \d t^{k-1} } x(t), \cdots , c_k \frac{ \d }{ \d t } x(t) \right).
\end{align*}
In order to have $\ov{F} ( \ov{x}(t) , t ) = \frac{ \d }{ \d t } \ov{x}(t) $, we let
$$\ov{F}( \ov{x}(t) , t ) = \left( c_1 F \left( c_1^{-1} \ov{x}_1(t), c_2^{-1} \ov{x}_2(t), \cdots, c_{k}^{-1} \ov{x}_{k}(t), t \right), c_2 c_1^{-1} \ov{x}_1(t) , c_3 c_2^{-1} \ov{x}_2(t), \cdots, c_k c_{k-1}^{-1} \ov{x}_{k-1}(t) \right).$$
Now, we check that indeed
\begin{align*}
\frac{ \d }{ \d t } \ov{x}(t) = & ~ \left( c_1 \frac{ \d^k }{ \d t^k } x(t) , c_2 \frac{ \d^{k-1} }{ \d t^{k-1} } x(t), \cdots , c_k \frac{ \d }{ \d t } x(t) \right) \\
= & ~ \left( c_1 F \left( \frac{ \d^{k-1} }{ \d t^{k-1} } x(t), \frac{ \d^{k-2} }{ \d t^{k-2} } x(t), \cdots , x(t), t \right) , c_2 \frac{ \d^{k-1} }{ \d t^{k-1} } x(t), \cdots , c_k \frac{ \d }{ \d t } x(t) \right) \\
= & ~ \left( c_1 F \left( c_1^{-1} \ov{x}_1(t), c_2^{-1} \ov{x}_2(t), \cdots, c_{k}^{-1} \ov{x}_{k}(t), t \right), c_2 c_1^{-1} \ov{x}_1(t) , c_3 c_2^{-1} \ov{x}_2(t), \cdots, c_k c_{k-1}^{-1} \ov{x}_{k-1}(t) \right) \\
& ~ \ov{F}( \ov{x}(t) , t ).
\end{align*}
We bound the Lipschitz constant of function $\ov{F}$ by
\begin{align*}
\| \ov{F} ( y(t) , t ) - \ov{F} ( z(t) , t ) \| 
\leq & ~ c_1 \sum_{i=1}^k L_i c_i^{-1} \| y_i(t) - z_i(t) \| + \sum_{i=1}^{k-1} c_{i+1} c_i^{-1} \| y_i(t) - z_i(t) \| \\
= & ~ \sum_{i=1}^{k - 1} ( c_1 L_i c_i^{-1} + c_{i+1} c_i^{-1} ) \| y_i(t) - z_i(t) \| + c_1 L_k c_k^{-1} \| y_k(t) - z_k(t) \|.
\end{align*}
We choose $c_1 = 1$, $c_i = \sum_{ j = i }^k L_j^{(i-1)/j} + \frac{1}{\ov{T}^{i-1}}$, $\forall i \in \{2, \cdots, k\}$ where $\ov{T} = 4 \gamma_\varphi T$. Then we can calculate for each $i \in [k-1]$,
\begin{align*}
c_1 L_i c_i^{-1} + c_{i+1} c_i^{-1} 
= & ~ \frac{ L_i }{ \sum_{ j = i }^k L_j^{(i-1)/j} + \frac{1}{\ov{T}^{i-1}} } + \frac{ \sum_{ j = i + 1 }^k L_j^{i/j} + \frac{1}{\ov{T}^{i}} }{ \sum_{ j = i }^k L_j^{(i-1)/j} + \frac{1}{\ov{T}^{i-1}} } \\
= & ~ \frac{ L_i }{ L_i^{(i-1)/i} + \sum_{ j = i+1 }^k L_j^{(i-1)/j} + \frac{1}{\ov{T}^{i-1}} } + \frac{ \sum_{ j = i + 1 }^k L_j^{i/j} + \frac{1}{\ov{T}^{i}} }{ \sum_{ j = i }^k L_j^{(i-1)/j} + \frac{1}{\ov{T}^{i-1}} } \\
\leq & ~ L_i^{1/i} + \sum_{j=i+1}^k L_j^{1/j} + \frac{1}{\ov{T}} \\
= & ~ \sum_{j=i}^k L_j^{1/j} + \frac{1}{\ov{T}}
\end{align*}
For $i = k$, we have
\begin{align*}
c_1 L_k c_{k}^{-1} = \frac{L_k}{ L_k^{(k-1)/k} + \frac{1}{ \ov{T}^{k-1} } } \leq L_k^{1/k}.
\end{align*}
Thus,
\begin{align*}
\| \ov{F} ( y(t) , t ) - \ov{F}( z(t) , t ) \| \leq & ~ \left( \sum_{j=1}^k L_j^{1/j} + \frac{1}{\ov{T}} \right) \cdot \sum_{i=1}^k \| y_i(t) - z_i(t) \| 
=  \left( \sum_{j=1}^{k} L_j^{1/j} + \frac{1}{\ov{T}} \right) \| y(t) - z(t) \|.
\end{align*}
It gives the following Claim:
\begin{claim}
Function $\ov{F}$ has Lipschitz constant $\ov{L} = \sum_{j=1}^k L_j^{1/j} + \frac{1}{\ov{T}}$.
\end{claim}

Then, we consider the following first order ODE,
\begin{align*}
\frac{ \d }{ \d t } \ov{x}(t) = & ~ \ov{F} ( \ov{x}(t) , t ) ,\\
\ov{x}(0) = & ~ (c_1 v_{k-1} , c_2 v_{k-2} , \cdots, c_k v_0)
\end{align*}
Let $\ov{x}^*(t) \in \R^{k d}$ denote the optimal solution, then
\begin{align*}
\ov{x}^*(t) = \left( c_1 \frac{ \d^{k-1} }{ \d t^{k-1} } x^*(t) , c_2 \frac{ \d^{k-2} }{ \d t^{k-2} } x^*(t) , \cdots , c_k x^*(t) \right).
\end{align*}
Now, we prove that $\frac{ \d }{ \d t } \ov{x}^*(t)$ is approximate by some element in $\mathcal{V}^d$. Let $\ov{q} : \R \rightarrow \R^{k d} $ be defined as follows
\begin{align*}
\ov{q}(t) = & ~ \left( c_1 q^{(k)}(t) , c_2 q^{(k-1)}(t) , \cdots , c_k q^{(1)}(t) \right) .
\end{align*}
Then,
\begin{align*}
& ~ \left\| \ov{q} (t) - \frac{ \d }{ \d t } \ov{x}^*(t) \right\| \\
= & ~ \left\| \left( c_1 q^{(k)}(t) , c_2 q^{(k-1)}(t) , \cdots , c_k q^{(1)}(t) \right)
- \left( c_1 \frac{ \d^k  }{ \d t^k } x^*(t), c_2 \frac{ \d^{k-1} }{ \d t^{k-1} } x^*(t) , \cdots , c_k \frac{ \d }{ \d t } x^*(t) \right) \right\| \\
= & ~ \sum_{i=1}^k c_i \left\|  q^{ (k + 1 - i) }(t) - \frac{ \d^{ k + 1 - i } }{ \d t^{ k + 1 - i } } x^*(t) \right\| \\
\leq & ~ \sum_{i=1}^k c_i \frac{ \epsilon }{ \ov{T}^{ k + 1 - i } }
=  \frac{1}{\ov{T}} \underbrace{ \epsilon \sum_{i=1}^k \frac{ c_i }{ \ov{T}^{ k - i } } }_{ \ov{\epsilon} } .
\end{align*}

By the assumption on $T$, we have that $\gamma_\varphi \ov{L} T \leq 1/2$ and hence theorem~\ref{thm:first_order_ode} finds 
\begin{align*}
\ov{x}^{(N)}(t) = ( \ov{x}_{1}^{(N)}(t) , \ov{x}_{2}^{(N)}(t) , \cdots , \ov{x}_{k}^{(N)}(t) ) \in \R^{k d}
\end{align*}
 such that
\begin{align}\label{eq:bound_ov_x_kth_order}
\| \ov{x}^{(N)}(t) - \ov{x}^*(t) \| \leq 20 \gamma_\varphi \ov{\epsilon}
\end{align}

This implies that
\begin{align} \label{eq:error_x_higher}
\left\| c_i^{-1} \ov{x}^{(N)}(t) - \frac{ \d^{k-i} }{ \d t^{k-i} } x^*(t) \right\| \leq 20 c_i^{-1}\gamma_\varphi \ov{\epsilon}, \forall i \in [k].                                                                                                                        
\end{align}
To bound the last term, we show the following Claim:
\begin{claim}

Let $\ov{\epsilon} = \epsilon \sum_{i=1}^k \frac{ c_i }{ \ov{T}^{k-i} } $. If we choose $c_1 = 1$ and $c_i = \sum_{j=i}^k L_j^{ ( i - 1 ) / j } + \frac{1}{\ov{T}^{i-1}} $, $\forall i \in \{2, \cdots, k\}$, then
\begin{align*}
c_i^{-1} \ov{\epsilon} \leq \frac{ \epsilon }{ \ov{T}^{k-i} } (2k+1).
\end{align*}
\end{claim}
\begin{proof}
For each $i \in [k]$,
\begin{align*}
c_i^{-1} \ov{\epsilon} =  c_i^{-1} \epsilon \sum_{j=1}^k \frac{ c_j }{ \ov{T}^{k - j} } 
=  \frac{ \epsilon }{ \ov{T}^{k - i} } \sum_{j=1}^k \frac{ c_j c_i^{-1} }{ \ov{T}^{ i - j } }
\end{align*}
We can lower bound the term $\frac{ c_j c_i^{-1} }{ \ov{T}^{ i - j } }$ as follows,
\begin{align*}
\frac{c_j c_i^{-1} }{ \ov{T}^{i - j} } = & ~ \frac{ \sum_{l=j}^k L_l^{(j-1)/l} + \frac{1}{ \ov{T}^{j-1} } }{ \ov{T}^{i - j} ( \sum_{l = i}^k L_l^{ ( i -1 ) / l  } + \frac{1}{\ov{T}^{i-1}} ) } \\
= & ~ \frac{ \sum_{l=j}^k L_l^{(j-1)/l} + \frac{1}{ \ov{T}^{j-1} } }{   \sum_{l = i}^k \ov{T}^{i - j} L_l^{ ( i -1 ) / l  } + \frac{1}{\ov{T}^{j-1}}  } \\
\leq & ~ \frac{ \sum_{l=j}^k L_l^{(j-1)/l} + \frac{1}{ \ov{T}^{j-1} } }{ \frac{1}{\ov{T}^{j-1}}  } \\
= & ~ 1 + \sum_{l = j}^k ( L_l^{1/l} \ov{T} )^{j-1} 
\end{align*}
Therefore, we have
\begin{align*}
c_i^{-1} \ov{\epsilon} = & ~ \frac{ \epsilon }{ \ov{T}^{k - i} } \sum_{ j = 1 }^k \frac{ c_j c_i^{-1} }{ \ov{T}^{ i - j } } \\
\leq & ~ \frac{ \epsilon }{ \ov{T}^{k - i} } \sum_{ j = 1 }^k \left( 1 + \sum_{l=j}^k ( L_l^{1/l} \ov{T} )^{j - 1} \right) \\
\leq & ~ \frac{ \epsilon }{ \ov{T}^{k - i} } \left( 2 k + \sum_{j=1}^k \left( \sum_{l=1}^{k}  L_l^{1/l} \ov{T} \right)^{j} \right) \\
\leq & ~ \frac{ \epsilon }{ \ov{T}^{k - i} } \left( 2k + \sum_{j=1}^k (1/2)^j \right)
\end{align*}
where we used $\gamma_\varphi L T \leq 1/8$ at the end.
Thus we complete the proof of Claim.
\end{proof}

Now, using the claim to (\ref{eq:error_x_higher}), we have the error is
$$20 (2k + 1) \gamma_\varphi \frac{\epsilon}{\ov{T}^{k-i}} \leq 20 (2k + 1) \gamma_\varphi \frac{\epsilon}{T^{k-i}}.$$

To bound the number of iterations needed in the log term in Theorem
\ref{thm:first_order_ode}, we note that
\begin{align*}
\int_{0}^{T}\|\overline{F}(\overline{x}(0),s)\|ds & =c_{1}\left\Vert \int_{0}^{T}F(c_{1}^{-1}\overline{x}_{1}(0),c_{2}^{-1}\overline{x}_{2}(0),\cdots,c_{k}^{-1}\overline{x}_{k}(0),s)ds\right\Vert +T\cdot\sum_{i=1}^{k-1}c_{i+1}c_{i}^{-1}\|\overline{x}_{i}(t)\|\\
 & =\left\Vert \int_{0}^{T}F(\frac{d^{k-1}}{dt^{k-1}}x(0),\frac{d^{k-2}}{dt^{k-2}}x(0),\cdots,x(0),s)ds\right\Vert +T\cdot\sum_{i=1}^{k-1}c_{i+1}\left\Vert \frac{d^{k-i}}{dt^{k-i}}x(0)\right\Vert .
\end{align*}
Note that
\[
c_{i+1}\leq\sum_{j=1}^{k}L_{j}^{i/j}+\frac{1}{\overline{T}^{i}}\leq L^{i}+\frac{1}{\overline{T}^{i}}\leq\frac{2}{\overline{T}^{i}}
\]
where we used $L\overline{T}\leq\frac{1}{2}$.
Hence, the number of iterations we need is
\begin{align*}
 & ~ O\left(D\log\left( \frac{1}{ \ov{\epsilon} } \left( \left\Vert \int_{0}^{T}F(v_{k-1},v_{k-2},\cdots,v_{0},s)ds\right\Vert +\sum_{i=1}^{k-1}\frac{\left\Vert v_{i}\right\Vert }{\overline{T}^{k-i-1}} \right) \right)\right)\\
= & ~ O\left(D\log\left( \frac{1}{ \epsilon } \left( \overline{T}^{k-1}\cdot\left\Vert \int_{0}^{T}F(v_{k-1},v_{k-2},\cdots,v_{0},s)ds\right\Vert +\sum_{i=1}^{k-1}\overline{T}^{i}\left\Vert v_{i}\right\Vert \right) \right)\right)
\end{align*}
where we used $\ov{\epsilon}\geq\epsilon\cdot c_{k}\geq\frac{\epsilon}{\overline{T}^{k-1}}$.
\end{proof}
\section{Deferred Proof for Cauchy Estimates (Section \ref{sec:low_degree_solutions})}\label{sec:app_logistic}

In this section, we provide the proofs of some core Lemmas/Claims used for proving Theorem~\ref{thm:bound_2ODE}. 







\restate{lem:bound_D_k_F_u}

\begin{proof}
Theorem \ref{thm:cauchy_kowalevski} shows that the solution $u$ uniquely exists and is real analytic around $0$. Therefore, we can take derivatives on both sides of $u''(t) = F(u(t))$ and get
\begin{align*}
u^{(3)} (t) = & ~ D F ( u(t) ) [ u^{(1)} (t) ], \\
u^{(4)} (t) = & ~ D F ( u(t) ) [ u^{(2)} (t) ] + D^2 F( u(t) ) [ u^{(1)} (t) , u^{(1)} (t) ] \\
u^{(5)} (t) = & ~ D F ( u(t) ) [ u^{(3)} (t) ] + 2 D^2 F( u(t) ) [ u^{(2)} (t) , u^{(1)} (t) ] + D^3 F ( u(t) ) [ u^{(1)}(t) , u^{(1)}(t) , u^{(1)}(t) ] \\
\vdots = & ~ \vdots 
\end{align*}

Therefore, we have
\begin{align*}
\| u^{(3)} (0) \| = & ~ \| D F ( u(0) ) [ u^{(1)} (0) ] \| \\
\leq & ~ f^{(1)}(0) \| u^{(1)} (0) \| \\
\| u^{(4)} (0) \| = & ~ \| D F ( u(0) ) [ u^{(2)} (0) ] + D^2 F( u(0) ) [ u^{(1)} (0) , u^{(1)} (0) ] \| \\
\leq & ~ \| D F ( u(0) ) [ u^{(2)} (0) ] \| + \| D^2 F( u(0) ) [ u^{(1)} (0) , u^{(1)} (0) ] \| \\
\leq & ~ f^{(1)} (0) \| u^{(2)} (0) \| + f^{(2)}(0) \| u^{(1)} (0) \|^2 \\
\| u^{(5)} (0) \| = & ~ \| D F ( u(0) ) [ u^{(3)} (0) ] + 2 D^2 F( u(0) ) [ u^{(2)} (0) , u^{(1)} (0) ] \\
& ~ + D^3 F ( u(0) ) [ u^{(1)} (0) , u^{(1)} (0) , u^{(1)} (0) ] \| \\
\leq & ~ \| D F ( u(0) ) [ u^{(3)} (0) ] \| + \| 2 D^2 F( u(0) ) [ u^{(2)} (0) , u^{(1)} (0) ] \| \\
& ~ + \| D^3 F ( u(0) ) [ u^{(1)} (0) , u^{(1)} (0) , u^{(1)} (0) ] \| \\
\leq & ~ f^{(1)} (0) \| u^{(3)} (0) \| + 2 f^{(2)}(0) \| u^{(2)} (0) \| \| u^{(1)} (0) \| + f^{(3)} (0) \| u^{(1)} (0) \|^3 \\
\vdots = & ~ \vdots
\end{align*}

Similarly, Theorem \ref{thm:cauchy_kowalevski} shows that the solution $\psi$ uniquely exists and is real analytic around $0$.
By expanding $\psi''(t) = f( \psi(t) )$ at $t = 0$, we see that
\begin{align*}
\psi^{(3)} (0) = & ~ f^{(1)} (0) \psi^{(1)} (0) , \\
\psi^{(4)} (0) = & ~ f^{(1)} (0) \psi^{(2)} (0) + f^{(2)} (0) ( \psi^{(1)} (0) )^2 , \\
\psi^{(5)} (0) = & ~ f^{(1)} (0) \psi^{(3)} (0) + 2 f^{(2)} \psi^{(2)} (0) \psi^{(1)} (0) + f^{(3)} (0) ( \psi^{(1)} (0) )^3 , \\
\vdots = & ~ \vdots 
\end{align*}
Since $\| u^{(1)} (0) \| \leq b = \psi^{(1)} (0)$ , $\| u^{(2)} (0) \| = \| F( u(0) ) \| \leq f(0) = \psi^{(2)} (0)$.

For $k = 3,4,5,\cdots$, we have
\begin{align*}
\| u^{(3)} (0) \| \leq & ~ f^{(1)} (0) \| u^{(1)} (0) \| = f^{(1)} (0) \psi^{(1)} (0) = \psi^{(3)} (0) \\
\| u^{(4)} (0) \| \leq & ~ \psi^{(4)} (0) \\
\| u^{(5)} (0) \| \leq & ~ \psi^{(5)} (0) \\
\vdots \leq & ~ \vdots 
\end{align*}

Thus, we have that $\| u^{(k)} (0) \| \leq \psi^{(k)} (0)$ for all $k \geq 1$. 

\end{proof}


\restate{lem:solve_single_variable_ODE_better}
\begin{proof}
Let $\widetilde{\psi}(t)=\frac{1}{c}\left(1-\sqrt{1-\alpha t}\right)$
with $\alpha=\max(\sqrt{\frac{4}{3}ac},2bc)$. Note that 
\[
\widetilde{\psi}''(t)=\frac{\alpha^{2}}{4c(1-\alpha t)^{3/2}}=\widetilde{f}(\widetilde{\psi}(t))\quad\text{with}\quad\widetilde{f}(x)=\frac{\alpha^{2}}{4c(1-cx)^{3}}.
\]
Since $\alpha^{2}\geq\frac{4}{3}ac$, we have that
\[
\widetilde{f}^{(k)}(0)=\frac{a}{3}\cdot\frac{(k+2)!}{2}c^{k}\geq k!\cdot a \cdot c^{k}=f^{(k)}(0).
\]
Also, we have that $\widetilde{\psi}(0)=0$ and $\widetilde{\psi}'(0)=\frac{\alpha}{2c}\geq b.$
Hence, Lemma \ref{lem:bound_D_k_F_u} shows that 
\[
\psi^{(k)}(0)\leq\widetilde{\psi}^{(k)}(0)=\frac{\alpha^{k}}{c}\prod_{i=1}^{k}\frac{|2i-3|}{2}\leq\frac{k!\alpha^{k}}{c}
\]
for all $k\geq1$.
\end{proof}

\section{Cauchy Estimates of Some Functions}\label{sec:app_cauchy}

We first state a useful tool,
\begin{lemma}\label{lem:holomorphic_boundary}
Let $f : U \rightarrow \mathbb{C}$ be holomorphic and suppose that $D_R = \{ z ~:~ |z - z_0 | \leq R \} \subset U$. Let $\gamma_R = \{ z ~:~ |z - z_0| = R \}$ denote the boundary of $D_R$. For all $n \geq 0$
\begin{align*}
| f^{(n)} ( z_0 ) | \leq \frac{ n! }{ R^n } \max_{z \in \gamma_R} |f(z)|.
\end{align*}
\end{lemma}

\subsection{Logistic loss function}

\begin{lemma}[Property of Logistic loss function]\label{lem:property_of_logisitc_function}
Let $\phi(t) = \log (1+ e^{-t})$, then we know that $\phi'(t)$  has Cauchy estimate $M=1$ with radius $r=1$.
\end{lemma}

\begin{proof}
Given the definition of $\phi(t)$, it is easy to see that
\begin{align*}
\phi'(t) = \frac{-1}{ 1 + e^{t} }.
\end{align*}

Let $f(z) = \frac{1}{1+e^{z}}$. Let $z = a + b \i$. We have
\begin{align*}
|f(z)|
=  \left| \frac{1}{1 + e^{a} \cos b + \i e^{a} \sin b } \right|
=  \frac{1}{ \sqrt{ (1 + e^{a} \cos b)^2 + ( e^{a} \sin b )^2 } } 
=  \frac{1}{ \sqrt{ 1 + 2 e^{a} \cos b + e^{2a} } }  
\end{align*}
Let $z_0 = a_0 + \i b_0$. We choose $R = 1$, then $(a-a_0)^2 + (b-b_0)^2 =1$. Since we only care about real numbers, we have $b_0 = 0$. Then we know that $b \in [-1,1]$, which means $\cos b \in [0.54, 1]$. Thus, we have,
\begin{align*}
\sqrt{ 1 + 2 e^{a} \cos b + e^{2a} } \geq \sqrt{ 1 + e^{a} + e^{2a} } \geq 1.
\end{align*}
Thus, $\left| \frac{1}{1 + e^{z}} \right| \leq 1$.
Therefore, using Lemma~\ref{lem:holomorphic_boundary}, for all real $z_0$,
\begin{align*}
|f^{(n)} (z_0)| \leq n!
\end{align*}

\end{proof}

\subsection{Pseudo-Huber loss function}

Huber function \cite{h64} has been extensively studied in a large number of algorithmic questions, e.g. regression \cite{cw15soda,cw15focs}, low-rank approximation \cite{swz19}, clustering \cite{bfl16}, sparse recovery \cite{nsw19}. Formally speaking, the Huber loss function can be defined as 
\begin{align*}
f(x) = 
\begin{cases}
\frac{x^2}{2 \delta}, & \text{~if~} |x| \leq \delta; \\
|x| - \delta /2 , & \text{~otherwise~}.
\end{cases}
\end{align*}
where $\delta > 0$ is a parameter. For many applications, the Pseudo-Huber loss function can be used as a smooth alternative for the Huber loss function.

\begin{lemma}[Property of Pseudo-Huber function]
Fix any $\delta>0$. Let $\phi(x) = \sqrt{ x^2 + \delta^2 } - \delta$, then we know that $\phi'(x)$ has Cauchy estimate $M=1$ with radius $r= \delta/2$.
\end{lemma}
\begin{proof}
Given the definition of function $\phi(x)$, it is easy to see that
\begin{align*}
\phi'(x) = \frac{ x }{ \sqrt{ x^2 + \delta^2 } }
\end{align*}
Let $f(z) = \frac{ z }{ \sqrt{ z^2 + \delta^2 } }$. Let $z = a + b \i$. We have
\begin{align*}
|f(z)| 
=  ~ \left| \frac{ (a+b \i) }{ \sqrt{ (a + b \i)^2 + \delta^2 } } \right| 
=  ~ \left| \frac{ (a+b \i) }{ \sqrt{ a^2 - b^2 + \delta^2 + 2 ab \i } } \right| 
=  ~ \frac{  |(a+ b \i)| }{ | \sqrt{a^2 - b^2 + \delta^2 + 2 a b \i} | }
\end{align*}
For the numerator, we have $| (a + b \i)| = \sqrt{a^2 + b^2}$. For the denominator, we have
\begin{align*}
| \sqrt{a^2 - b^2 + \delta^2 + 2 a b \i} | = \left( (a^2 - b^2 + \delta^2)^2 + 4 (ab)^2 \right)^{1/4}
\end{align*}

Let $z_0 = a_0 + \i b_0$. We choose $R = \frac{\delta}{2}$, then $(a-a_0)^2 + (b-b_0)^2 =(\delta/2)^2$. Since we are only real $z_0$, we have $b_0 = 0$. Then we know that $b \in [-\frac{\delta}{2},\frac{\delta}{2}]$. Thus
\begin{align*}
|f(z)| = & ~ \frac{ \sqrt{a^2 + b^2} }{ \left( (a^2 - b^2 + \delta^2)^2 + 4 (ab)^4 \right)^{1/4} } 
\leq  ~ \frac{ \sqrt{a^2 + \frac{\delta^2}{4}} }{ \left( (a^2 - \frac{\delta^2}{4} + \delta^2)^2 + 0 \right)^{1/4} } 
=  ~ \frac{ \sqrt{a^2 + \frac{\delta^2}{4}} }{ \sqrt{a^2 + \frac{3 \delta^2}{4}} } 
\leq  ~ 1.
\end{align*}
Therefore, using Lemma~\ref{lem:holomorphic_boundary}, for all real $z_0$,
\begin{align*}
|f^{(n)} (z_0)| \leq n! \cdot (2/\delta)^{n}.
\end{align*}
\end{proof}


\end{document}